\documentclass[lettersize,journal]{IEEEtran}

\usepackage{multirow}
\usepackage[textwidth=1cm,backgroundcolor=yellow,textsize=tiny]{todonotes}
\newcommand{\itamar}[1] {\todo[color=blue!20,inline]{Itamar: #1}}
\newcommand{\paolo}[1]{\todo[color=green!20,inline]{Paolo: #1}}

\newcommand{\ignore}[1]{}

\usepackage{cite}
\ifCLASSINFOpdf
\else
\fi

\usepackage{amsmath,amsfonts, amsthm, amssymb, bm, color, xspace, paralist, xcolor}
\usepackage{textcomp, stfloats, verbatim}
\usepackage{url}
\PassOptionsToPackage{hyphens}{url}
\usepackage{graphicx}
\newtheorem{theorem}{Theorem}

\newtheorem{property}{Property}
\newtheorem{proposition}[theorem]{Proposition}

\newcommand{\example}[2]{
        \rule{0.95\columnwidth}{0.5mm}\\
        \noindent {\bf Example:~#1}
        #2\\
        \rule{0.98\columnwidth}{0.5mm}
}

\newcommand{\abs}[1]{\left\vert#1\right\vert}
\newcommand{\set}[1]{\left\{#1\right\}}

\usepackage{mathtools}

\DeclarePairedDelimiter\ceil{\lceil}{\rceil}

\usepackage{algorithm}
\usepackage{comment, tagging}
\usepackage[noend]{algpseudocode}
\algnewcommand{\AND}{\textbf{and}\xspace}
\algnewcommand{\OR}{\textbf{or}\xspace}
\algnewcommand\algorithmicforeach{\textbf{for each}}
\algdef{S}[FOR]{ForEach}[1]{\algorithmicforeach\ #1\ \algorithmicdo}

\algrenewcommand\algorithmicindent{1.0em}%

\makeatletter
\algnewcommand{\LineComment}[1]{\Statex \hskip\ALG@thistlm \(\triangleright\) #1}
\makeatother

\newbox\statebox
\newcommand{\myState}[1]{%
    \setbox\statebox=\vbox{#1}%
    \edef\thealgruleheight{\dimexpr \the\ht\statebox+1pt\relax}%
    \edef\thealgruledepth{\dimexpr \the\dp\statebox+1pt\relax}%
    \ifdim\thealgruleheight<.75\baselineskip
        \def\thealgruleheight{\dimexpr .75\baselineskip+1pt\relax}%
    \fi
    \ifdim\thealgruledepth<.25\baselineskip
        \def\thealgruledepth{\dimexpr .25\baselineskip+1pt\relax}%
    \fi
    \State #1%
    \def\thealgruleheight{\dimexpr .75\baselineskip+1pt\relax}%
    \def\thealgruledepth{\dimexpr .25\baselineskip+1pt\relax}%
}

\makeatletter
\def\blfootnote{\xdef\@thefnmark{}\@footnotetext}
\makeatother

\usepackage{float, graphicx, tabu, color, tabto} %
\usepackage[caption=false]{subfig}

\usepackage{stackengine} 

\usepackage{array}
\newcolumntype{P}[1]{>{\centering\arraybackslash}p{#1}}
\newcolumntype{M}[1]{>{\centering\arraybackslash}m{#1}}
\usepackage{setspace} 							 

\newcommand{\Gc}{\mathcal G}

\newcommand{\Lc}{\mathcal L}
\newcommand{\Mc}{\mathcal M}

\newcommand{\Pc}{\mathcal P}

\newcommand{\Sc}{\mathcal S}
\newcommand{\Rc}{\mathcal R}
\newcommand{\Tc}{\mathcal T}

\newcommand{\mv}{m}

\newcommand{\yv}{\mathbf y}

\newcommand{\maxCpuPerService}{\beta^{r,s}_M}
\newcommand{\alg} {DASDEC}
\newcommand{\migProb}{PMP}
\newcommand{\cpuall}{GFA}
\newcommand{\bu}{SFS} 
\newcommand{\fMode}{F} 
\newcommand{\buFMode}{\fMode-\bu}

\newcommand{\pd}{PD}

\DeclareMathOperator{\buMath}{SFS}
\DeclareMathOperator{\pdMath}{\pd}

\newcommand{\pu}{PU}

\newcommand{\lBound}{LBound}
\newcommand{\notAssigned}{\tilde{\Mc}} 
\newcommand{\notAssignedOfS}{\notAssigned_s} 
\newcommand{\pushUpSet}{\Pc^U} 
\newcommand{\pushUpSetOfS}{\pushUpSet_s} 
\newcommand{\pushDownSet}{\Pc^D} 
\newcommand{\pushDownSetOfS}{\pushDownSet_s} 

\newcommand{\ms}{MultiScaler}
\newcommand{\deficitCpu}{deficitCpu}
\DeclareMathOperator{\deficitCpuMath}{\deficitCpu}

\newcommand{\accDelayBu}{T^{\buMath}_{ad}}
\newcommand{\accDelayPd}{T^{\pdMath}_{ad}}
\newcommand{\ffit}{F-Fit}
\newcommand{\cpvnf}{CPVNF}

\newcommand{\subFigWidth}{0.98 \columnwidth}

\newcommand{\algFontSize}{\scriptsize}

\begin{document}
\title{Distributed Asynchronous Service Deployment in Edge-Cloud Multi-Tier Networks}

\author{Itamar Cohen,~\IEEEmembership{Member,~IEEE} Paolo~Giaccone,~\IEEEmembership{Senior~Member,~IEEE}, Carla Fabiana Chiasserini,~\IEEEmembership{Fellow,~IEEE}
\IEEEcompsocitemizethanks{\IEEEcompsocthanksitem
The authors are with Politecnico di Torino; e-mail: \{firstname.lastname\}@polito.it. C.\,F.\, Chiasserini and P.~Giaccone are also with CNIT, Parma, Italy, and CNR-IEIIT, Torino, Italy.}
\thanks{This work was supported by the European Union’s NextGenerationEU instrument, under the Italian National Recovery and Resilience Plan (NRRP), M4C2 Investment 1.3, “Telecommunications of the Future” (PE00000001), program “RESTART”, and by the European Commission through Grant No. 101095890 (Horizon Europe SNS JU PREDICT-6G project).}}

\maketitle

\begin{abstract}
In an edge-cloud multi-tier network, datacenters provide services to mobile users, with each service having specific 
latency and computational requirements. Deploying such a variety of services while matching their requirements with the available computing resources is challenging. In addition, time-critical services may have to be migrated as the users move, to keep fulfilling their latency constraints.  
Unlike previous work relying on an orchestrator with an always-updated global view of the available resources and the users' locations, this work envisions a distributed solution to the above problems. 
In particular, we propose a distributed asynchronous framework for service deployment in the edge-cloud that increases the system resilience by avoiding a single point of failure, as in the case of a central orchestrator. 
Our solution ensures cost-efficient feasible placement of services, with negligible bandwidth overhead.
Our results, obtained through trace-driven, large-scale simulations, 
show that the proposed solution provides performance very close to those obtained by state-of-the-art centralized solutions, and at the cost of a small communication overhead.
\end{abstract}


{\section{Introduction}\label{sec:introduction}}
\IEEEPARstart{T}{oday's} networks offer extensive virtualized resources, structured as a collection of datacenters across the different network layers from the edge to the cloud~\cite{FollowMe_J, micado_orchestrator, SFC_mig, orch_cloud2edge_survey, MoveWithMe, Justify_CLP_tree_ISPs}.
These datacenters host a plethora of applications with versatile computational requirements and latency constraints. For example, time-critical services such as road safety applications require low latency, necessitating processing in an edge datacenter, close to the user. In contrast, infotainment tasks may require larger computational resources, but have looser latency constraints and therefore may be placed on a cloud datacenter with abundant and affordable computation resources~\cite{tong2016_justify_path_to_root, SFC_mig}. 
Deploying services in an edge-cloud multi-tier network while considering computational and network resources is challenging. The complexity increases when user mobility or fluctuating traffic demands necessitate migrating services to maintain low latency and meet KPI constraints~\cite{FollowMe_J,MoveWithMe, mig_correlated_VMs}.				   
Additionally, some services may require migration as users move to maintain latency constraints~\cite{FollowMe_J,MoveWithMe_short, mig_correlated_VMs}.

Providing stable service to such heterogeneous mobile users is a complex problem, akin to dynamic bin-packing, with conflicting optimization goals. Indeed, meeting the latency constraints incentivizes deploying more services on the edge, closer to the user. At the same time, to reduce cost and minimize the number of future migrations, one should deploy more services in the cloud, where resources are cheaper. Finally, if a user with tight delay constraints moves towards an edge datacenter that does not have enough available resources, some services already deployed on it should be migrated to make room for the new service. 

{\bf Research gaps.} 
Most existing solutions~\cite{FollowMe_J,Justifies_path_to_root_n_CLP_vehs, Dynamic_Service_Provisioning_ToN, CPVNF_proactive_place_in_CDN, micado_orchestrator, SFC_mig_mechanism, orch_cloud2edge_survey, MultiScaler_Kassler, dynamic_sched_and_reconf_t, agarwal-19, Avatar, SFC_mig, Dynamic_SFC_by_rtng_Dijkstra}
rely on a central orchestrator for service deployment and migration decisions.
The orchestrator periodically (i) gathers information about the status of the resources and the migration requirements, (ii) calculates new placement and resource allocation, and (iii) instructs datacenter local controllers accordingly.
This centralized, synchronous approach has several shortcomings. {\em First}, gathering fresh state information causes significant communication bottlenecks and a severe incast congestion at the orchestrator's input, even within a single cloud datacenter~\cite{APSR}. Hence, the centralized approach does not scale well, thus failing to manage systems with multiple datacenters efficiently. 
As a result, in an actual implementation, the entire system state can only be roughly estimated. However, inaccurate state estimation may significantly deteriorate performance~\cite{Obl_vs_stateful_edge-cloud_Marco}. {\em Second}, the orchestrator is a natural single point of failure, compromising the system's reliability. {\em Third}, the datacenters may be operated by distinct operators~\cite{Crosshaul, NFV_ego_learning}, which are typically unwilling to share proprietary information and implementation details with competitors. 

\textbf{Our contribution.} 
We present the Distributed Asynchronous Service Deployment in the Edge-Cloud  (\alg) framework, which minimizes the migration and placement costs.
In \alg, a user experiencing increasing latency beyond an acceptable threshold triggers a migration request from the edge nodes to datacenters to which the corresponding service may be migrated to meet the latency constraints. As a result, \alg\
\begin{itemize}
\item shows a negligible communication overhead by using simple control messages transmitted by a datacenter to only relevant neighboring datacenters;
\item runs asynchronously across distinct datacenters without requiring global synchronization; 
\item can find a deployment solution even with small available resources in the network for the services; 
\item allows multiple datacenters -- possibly of distinct providers -- to cooperate without exposing proprietary information. 
Indeed, it assumes that each datacenter knows only its available resources, and does not require such proprietary information to be shared with other datacenters.
\end{itemize} 

\textbf{Paper organization.}		   
We present the system model in Sec.~\ref{sec:system}. Sec.~\ref{sec:problem} formulates the placement and migration problem and describes the algorithmic challenges. Sec.~\ref{sec:alg} details the \alg\ algorithmic framework. Sec.~\ref{sec:sim} evaluates the performance of \alg\ in various settings.
Finally,  we review some relevant related work in Sec.~\ref{sec:related_work} and draw conclusions in Sec.~\ref{sec:conclusion}.

\section{System Model}\label{sec:system}

This section formally defines our system model and notations, which are summarized in Table~\ref{tab:notations}.
\begin{figure}
    \centering
    \includegraphics[width=\columnwidth]{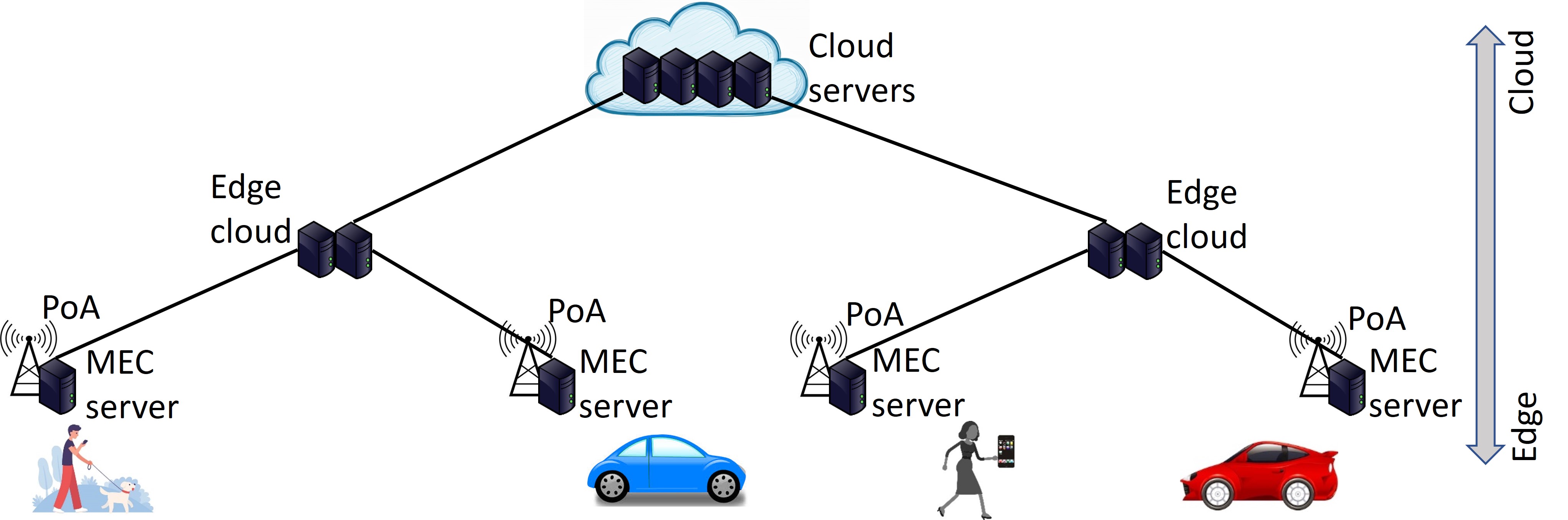}
    \caption{    
    Reference scenario: Each mobile user (illustrated by a car or a human) is associated with the nearest PoA, which is equipped with a co-located Mobile Edge Computing (MEC) server. The MEC servers are connected to higher-level servers with more computational capacity. Our model refers to such servers in the edge-cloud continuum as datacenters.    
    }
    \label{fig:model}
\end{figure}

We consider a fat-tree edge-cloud multi-tier network  architecture, which comprises~\cite{tong2016_justify_path_to_root}: 
\begin{inparaenum}[(i)]
\item {\em datacenters}, denoting generic computing resources,
\item {\em switches}, and
\item {\em Points of Access (PoA)}, each equipped with a radio interface and connected with an edge datacenter~\cite{Mig_in_Mobile_Edge_Clouds}. 
\end{inparaenum}
Datacenters are interconnected via switches, while users connect to the network through a PoA, which may change as they move. 
By denoting the set of {\em datacenters} with $\Sc$, we model the logical multi-tier network under study  as a directed graph ${\Gc}{=} (\Sc,\Lc)$, where   
the edges are the directed virtual links connecting the datacenters.
As done  in~\cite{Dynamic_Service_Provisioning_ToN, Justify_CLP_tree_ISPs, tong2016_justify_path_to_root}, We assume a single predetermined loop-free path exists between each pair of datacenters, and denote with $\Tc(s)$ the sub-tree rooted by datacenter $s$. 
When referring to datacenter positions (i.e., using terms such as up, down, below, above, etc.), we consider a network topology where the root is at the top of the topology and the leaves are at the bottom. 	

We consider a generic user generating a {\em service request} $r$, originating at the PoA $p^r$, to which the user is currently connected. 
The request $r$ can be served by placing service instance $\mv^r$ on a datacenter.
Let $\Rc$ denote the set of service requests, and $\Mc$ the set of corresponding services that are currently placed, or need to be placed, on datacenters. An example of the system is depicted in Fig.~\ref{fig:model}.

Each service is associated with an SLA, which specifies its key performance indicator (KPI) constraints~\cite{Okpi}.
Let us consider latency the most relevant KPI, although our model could also be extended to other KPIs, e.g.,  throughput.
Each service request $\mv^r$ is subject to the latency constraint $d^{\max}_{\mv^r}$, which is the maximum acceptable end-to-end round-trip delay to access the datacenter serving $\mv^r$.
Due to its latency constraint, each request $r$ is associated with a list of {\em latency-feasible} datacenters $\Sc_r$. 
The latency-feasible datacenters in $\Sc_r$ are not too far from $r$'s PoA ($p^r$), 
located along the path from the  $p^r$ to the root, coherently with~\cite{Justifies_path_to_root_n_CLP_vehs, Justify_CLP_tree_ISPs, Dynamic_Service_Provisioning_ToN}.
The top latency-feasible datacenter of request $r$ is denoted by $\bar{s}_r$. 

To successfully serve request $r$ on datacenter $s \in \Sc_r$ while satisfying the KPI constraint, $s$ should allocate (at least) $\beta^{r,s}$ CPU units, where $\beta^{r,s}$ is an integer multiple of a basic CPU speed. 
Since $\Sc_r$ and $\beta^{r,s}$ can be easily computed given the network and service characteristics~\cite{Dynamic_Service_Provisioning_ToN}, we consider  $\Sc_r$ and $\beta^{r,s}$ as known input parameters. 
Each datacenter $s\in \Sc$ has a total processing capacity $C_s$, expressed in the number of CPU cycles/s.
											
\begin{table}
\centering
    \small
    \caption{Main notations
    \label{tab:notations}}
    \begin{tabular}{|M{0.17\columnwidth}|m{0.71\columnwidth}|}
        \hline 
	    {\bf Symbol} & \multicolumn{1}{c|} {\bf Description} \tabularnewline
	    \hline 
	    \multicolumn{2}{|c|}{System model (Sec.~\ref{sec:system})
	    }\\
        \hline
        $\Gc$ & Network graph \tabularnewline
        $D(\Gc)$ & Diameter  of network graph $\Gc$ \tabularnewline
        $\Mc$ & Set of currently active services \tabularnewline
        $\Rc$ & Set of service requests \tabularnewline
        $p^r$ & Point of Access where request $r$ is generated \tabularnewline
        $\Sc$ & Set of datacenters \tabularnewline
        $\Tc(s)$ & Sub-tree rooted by datacenter $s$\tabularnewline
        $\mv^r$ & Service associated with request $r$ \tabularnewline
        $d^{\max}_{\mv^r}$ & latency constraint for service $\mv^r$ \\
        $\beta^{r,s}$ & Minimal number of CPU units required to place service $\mv^r$ on datacenter $s$ \tabularnewline
        $\psi^c(r,s)$ & Cost of placing service $r$ on datacenter $s$\tabularnewline
        $\psi^m(u, s, s')$ & Cost of migrating service $\mv^r$  from datacenter $s$ 
        to datacenter $s'$\tabularnewline
        $C_s$ & Overall processing capacity of datacenter $s$ [CPU cycles/s] \tabularnewline
        $\Sc_r$ & Set of latency-feasible datacenters for request $r$ \tabularnewline
        $\overline{s}_r$ & Top datacenter in $\Sc_r$\tabularnewline
	    \hline 
	    \multicolumn{2}{|c|}{\alg\ (Sec.~\ref{sec:alg})
	    }\\
        \hline
        $x (u,s)$ & Placement indicator: equal to 1 iff service $\mv^r$ is currently hosted on datacenter $s$\tabularnewline
        $y (r,s)$ & Boolean decision variable: equal to 1 iff service $\mv^r$ is scheduled to run on datacenter $s$\tabularnewline
        $\phi$ & Objective function~\eqref{Eq:def_obj_func} \tabularnewline
        $\pushUpSetOfS$ & Set of assigned services to be pushed up\tabularnewline
        $\pushDownSetOfS$ & Set of services to be pushed down.\tabularnewline   
        $a_s$ & Available processing capacity of datacenter $s$ \tabularnewline   
        $\notAssigned$ & Set of currently unassigned services\tabularnewline   
        \hline
        \end{tabular}
\end{table}

\section{The Placement and Migration Problem}\label{sec:problem}
											   
We start by observing that the latency experienced by a service may vary over time due to either
\begin{inparaenum}[(i)]
\item a change in the user's PoA (due to mobility), which changes the experienced {\em network latency}, or
\item a fluctuation in the service traffic, and, hence,  in the corresponding {\em processing latency} at the datacenter hosting the service~\cite{Dynamic_user_demands}.
\end{inparaenum}
The service client at the user continuously monitors 
its Quality of Experience (QoE) and warns its PoA as the experienced latency approaches the constraint\footnote{If the user can predict its near-future location, it can inform the network when the latency constraint is {\em about} to be violated.}. 
The user is then tagged as {\em critical},  i.e., the latency constraint is about to be violated, and the PoA triggers a {\em migration process} for it.  
Furthermore, the PoA is in charge of {\em placing instances of services} that have been newly requested by users under its control. The migration of service instances corresponding to critical users (hereinafter also referred to as critical services) may also involve the migration of other already-deployed non-critical services, whenever this is necessary (i) to successfully place newly arrived service requests {\em and} (ii) satisfy the latency constraints of all -- new and already deployed -- service instances. 

We now formalize the {\em Placement and Migration Problem (PMP)}, which aims to minimize the cost of overall service migration and placement operations.

{\bf Decision variables.} 
A placement strategy must accommodate new requests and those that become critical due to user mobility.
Let {matrix} $\yv$ 
represent the Boolean placement decision variables, i.e., $y(r,s) {=} 1$ if service $\mv^r$ is scheduled to run on datacenter $s$. 
Any choice for the values of such variables provides a {\em solution} to the PMP, determining (i) where to deploy new services, (ii) which existing services to migrate, and (iii) where to migrate them.

A {\em synchronous} placement algorithm should determine the placement of all new and critical requests once in a pre-defined time slot (e.g., once in a second). 
An {\em asynchronous} algorithm, on the other hand, looks for placement of each request as soon as it arrives (for a new request), or becomes critical (for an existing request). 

{\bf Constraints.} 
The following conditions must be met:
\begin{align}
\textstyle
\sum_{s \in \Sc} y(r,s)
&= 1
&\forall r \in \Rc \label{problem_constraint:single_placement} \\
\textstyle
\sum_{r \in \Rc} y(r,s) \beta^{r,s}
&\leq C_s
&\forall s \in \Sc\,. \label{problem_constraint:datacenter_computational_residual_capacity}
\end{align}
Constraint \eqref{problem_constraint:single_placement} ensures that at any point in time, each service $\mv^r$ is associated with a single {scheduled} placement. Constraint \eqref{problem_constraint:datacenter_computational_residual_capacity}, instead, assures that the capacity of each datacenter is not exceeded while guaranteeing that the maximum processing latency required for a service instance is met. 
Note that~\eqref{problem_constraint:datacenter_computational_residual_capacity} also embodies the latency constraints, as a higher $\beta^{r,s}$ is required to deploy a service with tighter latency constraint; in particular, the case where deploying request $r$ on datacenter $s$ violates $r$'s latency constraint is captured by determining $\beta^{r,s}=\infty$.

{\bf Costs.}
Placing service $r$ on datacenter $s$ incurs a processing cost $\psi^c(r,s)$, which captures the computational resource usage for service execution and network resources for data transfer.
As cloud computational resources, as well as the bandwidth in the network connecting cloud datacenters, are cheaper~\cite{tong2016_justify_path_to_root, SFC_mig} than edge datacenters, we assume that, if $s$ is an ancestor of $s'$, placing service $\mv^r$ on $s$ is cheaper than placing $\mv^r$ on $s'$.
The overall processing cost across all the datacenter is given by:
\[
\sum_{s \in \Sc}     \sum_{r \in \Rc}     y(r,s) \psi^c(r,s)
\]

Migrating service $r$ from datacenter $s$ to data center $s'$ incurs a {\em migration cost} $\psi^m(r,s,s')$, representing the network resources usage, migration pre-copy/post-copy process. 
To describe the placement according to the current deployment, we define the binary variable  $x(r,s)$ defined equal to $1$ iff service $\mv^r$ is currently placed on datacenter $s$.  Notice that $x(r,s)$ is not a decision variable.

We assume that a service instance 
does not become critical again before it finishes being placed based on the decision made by the previous solution to the PMP. 
The migration cost incurred by a critical service $\mv^r$ is then given by: 
\begin{align}
\sum_{s\neq s' \in \Sc} x(r,s) \cdot y(r,s') \cdot \psi^m(r,s,s'). \notag
\end{align} 

{\bf Objective.}
The Placement and Migration Problem (\migProb) is to minimize the overall migration and processing costs:
    \begin{multline}\label{Eq:def_obj_func}
    \phi(\yv) =
    \sum_{{s\neq s' \in \Sc }} 
    \sum_{r \in \Rc} 
    x(r,s) \cdot y(r,s') \cdot \psi^m(r,s,s') \\
    + 
    \sum_{s \in \Sc} 
    \sum_{r \in \Rc} 
    y(r,s) \psi^c(r,s)
    \end{multline}
subject to constraints~\eqref{problem_constraint:single_placement}--\eqref{problem_constraint:datacenter_computational_residual_capacity}.

The following property 
captures the complexity of the \migProb. 
\begin{property}
\label{prop:hard}
The problem of minimizing \eqref{Eq:def_obj_func} subject to~\eqref{problem_constraint:single_placement}--\eqref{problem_constraint:datacenter_computational_residual_capacity} is NP-hard.
\end{property}
\begin{IEEEproof}
    Our proof is established through a reduction from the NP-hard minimal global assignment problem (MINGAP)~\cite{Book_GAP_and_even_knowing_if_exists_feas_is_NPH}. In MINGAP, the input comprises a list of items and knapsacks. The capacity of knapsack $j$ is $C_j$; assigning item $i$ into knapsack $j$ incurs cost $c_{ij}$ and weight $w_{ij}$. The objective is to find a feasible assignment of all the items to knapsacks, while minimizing the overall cost.
    
     We now present a polynomial reduction from the MINGAP to the \migProb.
    The processing capacity of each datacenter $C_s$ is identical to the size of the respective knapsack. 
 The number of CPU units required to place each service on a datacenter is $\beta$.     
    Hence, there exists a feasible solution for \migProb\ for this input iff there exists a solution to MINGAP.
\end{IEEEproof}

\subsection{Algorithmic Challenges}\label{sec:alg_concept}
We are interested in a \textit{distributed solution, where no single datacenter (or any other entity) has a fresh all-encompassing view of the resource status} (e.g., the current place of each service, or the amount of available resources in each datacenter). Conversely,  the computation for service  placement and migration  
should involve a subset of the datacenters as small as possible.
Furthermore, the solution should be \textit{asynchronous}, as distinct PoAs may independently invoke different, simultaneous algorithm runs. Developing such a distributed asynchronous solution for the \migProb\ poses the following challenges.  

{\em 1. Feasibility vs.\ cost efficiency.}
A cost-effective approach may seem to involve placing services closer to the cloud because (i) resources are cheaper there, and (ii) higher-tier placement may reduce future migrations. 
However, such an approach may prevent an algorithm from finding feasible solutions, even when they exist.  
For example, consider the scenario depicted in Fig.~\ref{fig:need_bu_problem}, where requests can be placed at every level of the tree, and the capacity of each datacenter is sufficient for only one service. 
Requests arrive in the order of $r_0, r_1, r_2, r_3$. Let us consider a naive algorithm that, according to the order of arrival of the requests $r_i$ ($i{=}0,\ldots,3$), places each requested service on its highest latency-feasible datacenter, which is still available at the time of arrival on $\Sc_{r_i}$.   As illustrated in Fig.~\ref{fig:need_bu_problem}, after $r_0$, $r_1$, and $r_2$ have been placed, request $r_3$ finds all its latency-feasible datacenters full, causing the naive algorithm to fail. Instead, Fig.~\ref{fig:need_bu_sol} shows that a feasible solution for the problem exists. This confirms that the placement strategy of always choosing the highest-available datacenter may fail to find a feasible solution. Hence, online algorithms must establish the proper priority between feasibility and cost efficiency. 

\begin{figure}
    \centering
    \subfloat[\label{fig:need_bu_problem}Failure of naive placement algorithm] {
     \includegraphics[height=2.0 cm]{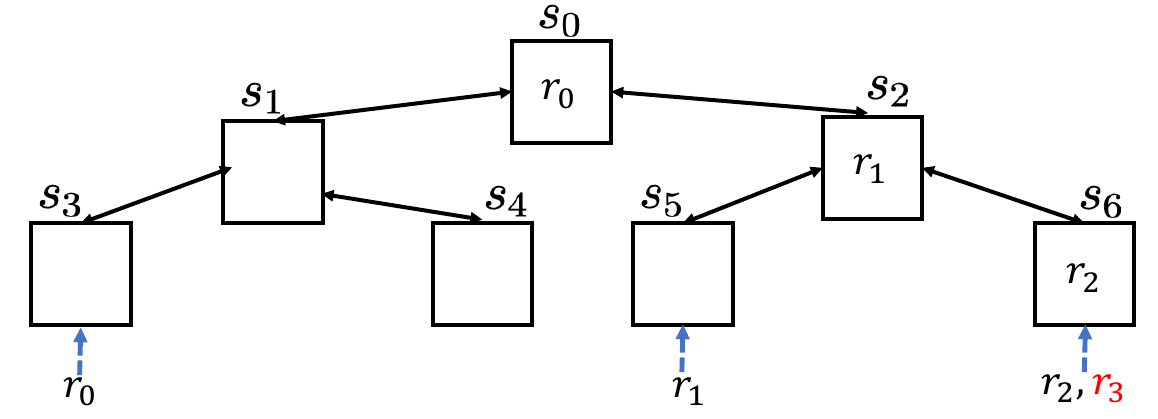}
    }
    
    \subfloat[\label{fig:need_bu_sol}Feasible placement] { 
    \includegraphics[height=2.0 cm]     {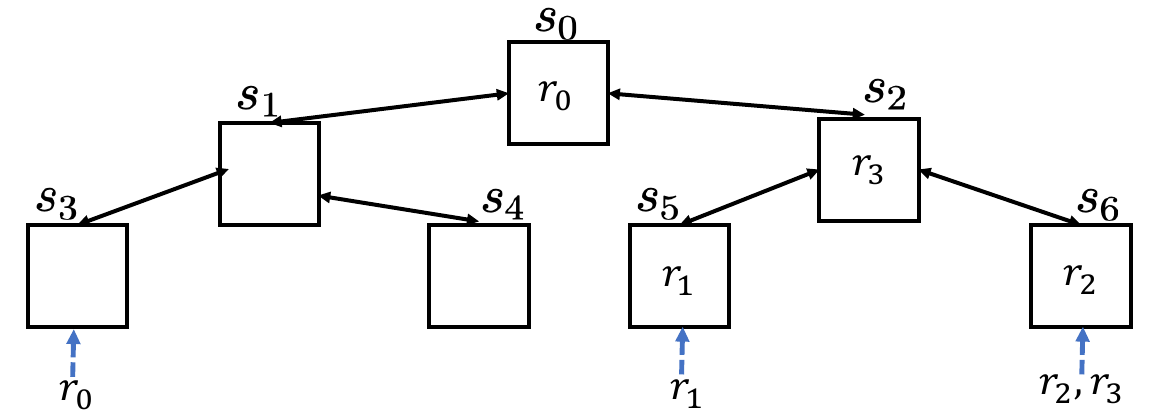} 
    }

    \subfloat[\label{fig:need_bu_r_3_finish}Feasible placement after $r_3$ finishes] { 
    \includegraphics[height=2.0 cm]     {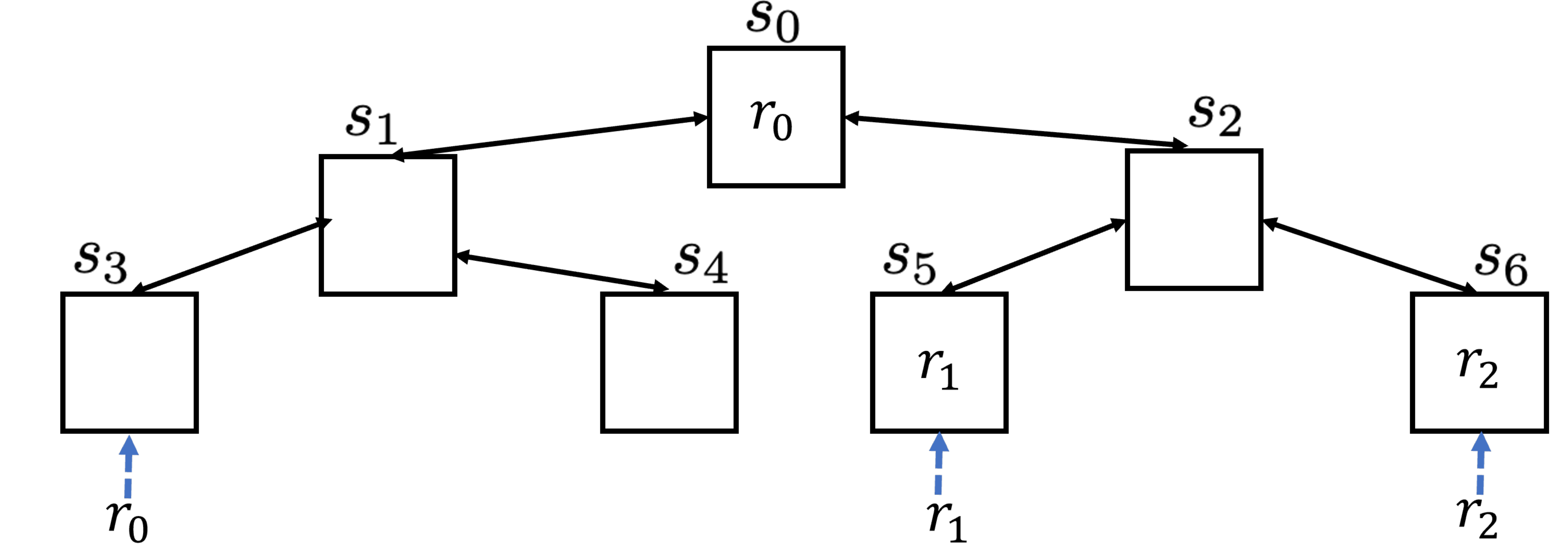} 
    }
    
    \subfloat[\label{fig:need_bu_after_pu}Feasible placement after pushing-up $r_1$] { 
    \includegraphics[height=2.0 cm]     {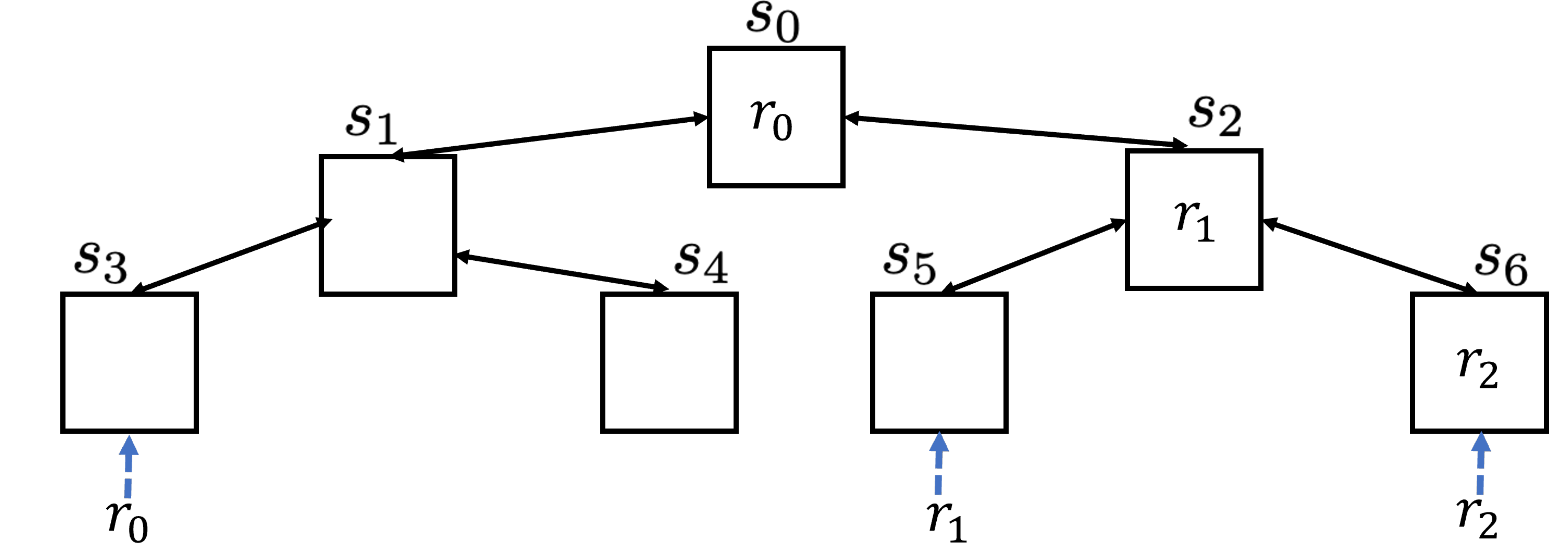} 
    }
    
    \subfloat[\label{fig:need_bu_fail_again}Fall back into a  scenario similar to the initial scenario (Fig.~\ref{fig:need_bu_problem})] { 
    \includegraphics[height=2.0 cm]     {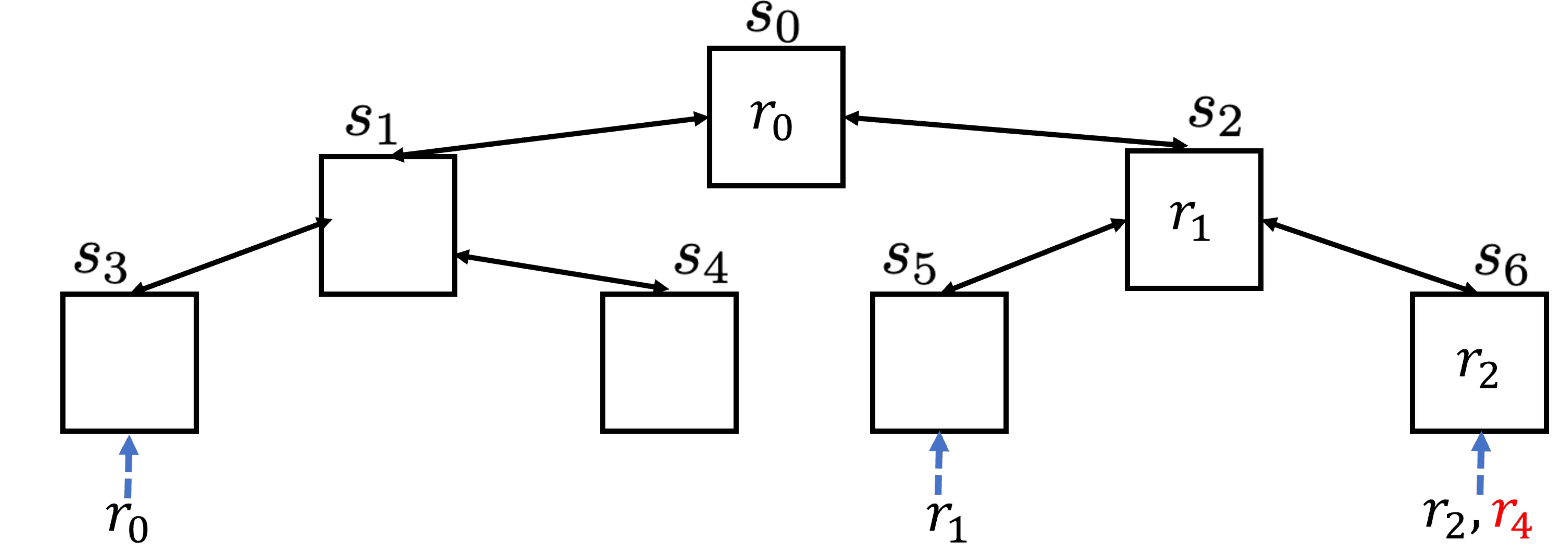} 
    }
    \caption{
            A scenario with seven datacenters (the cloud is at the root of the topology).         The PoAs of requests $r_0, r_1$, $r_2$,
            and $r_3$ are $s_3$, $s_5$, $s_6$, and $s_6$, respectively.
            A naive algorithm fails to place $r_3$ (Fig.~\ref{fig:need_bu_problem}), while there exists a feasible placement (Fig.~\ref{fig:need_bu_sol}). 
            After $r_3$ finishes (Fig.~\ref{fig:need_bu_r_3_finish}), an imprudent  algorithm may push request $r_1$ up again to datacenter $s_2$ (Fig.~\ref{fig:need_bu_after_pu}). This may result in falling back into a scenario that requires push-down (Fig.~\ref{fig:need_bu_fail_again}).
        }
    \label{fig:need_bu}
\end{figure}

{\em 2. Compulsory migration of non-critical services.}
In the online case, scenarios exist where it is necessary to migrate a non-critical service to make room for a critical service. 
For instance, suppose that the state of the placement 
is shown in Fig.~\ref{fig:need_bu_problem}. To place request $r_3$, the placement algorithm must migrate the non-critical service $r_1$ from $s_2$ to $s_5$. 
Migrating non-critical services in asynchronous settings, where new requests may arrive while the placement algorithm still works to make room for older service requests, is challenging. Indeed, an ill-planned migration scheme may result in races and even a nonconvergent algorithm.

{\em 3. Communication overhead.}
In distributed settings, datacenters must exchange information to ensure consistency—e.g., preventing multiple datacenters from simultaneously fulfilling the same service request. 
However, the communication overhead entailed by a solution strategy should be minimal to ensure that the amount of in-band control traffic exchanged is negligible compared to the data-plane service traffic. 

\section{The \alg\ Algorithmic Framework}\label{sec:alg}

We now present our algorithmic solution to \migProb, named  Distributed Asynchronous Service Deployment in the Edge-Cloud  (\alg). 
We start with a high-level description of the overall framework (Sec.~\ref{sec:alg:hi-lvl}) and then provide the details of the single algorithmic components (Sec.~\ref{sub:algo}). 
Next, we show how to tackle challenges that arise due to the asynchronous nature of the algorithm, namely, risks of races and non-convergence (Sec.~\ref{sec:alg:assure_converge}).  
Lastly, we consider the algorithm's communication overhead (Sec.~\ref{sec:reduce_commoh}).   

\subsection{Overview of the \alg\ framework}\label{sec:alg:hi-lvl}	 		

Each time a PoA identifies the existence of {\em not-assigned} requests -- namely, new and critical service requests -- the PoA triggers a run of \alg.
\alg\ tries to place each such request on a corresponding delay-feasible datacenter, with the following design principles, each of them corresponding to a specific stage in the algorithm:
\begin{enumerate}
\item Initially, seek a feasible solution by assigning each request to the lowest datacenter (i.e., closer to the edge) with enough resources.
\item Next, check for each such assigned service whether an ancestor datacenter can host it. If so, de-assign the service from the lower datacenter and let that ancestor place it. Otherwise, place the service on the lower datacenter.
\item If services that could not be assigned exist, try to free resources for these services by relocating services to lower datacenters.
\end{enumerate} 

Being distributed and asynchronous, \alg\ may run simultaneously across multiple datacenters. Distinct datacenters run local instances of the algorithm, composed of three main stages, as detailed below.

\textbf{Stage 1: Seek Feasible Solution (\bu).} 
This stage is triggered when a PoA detects unassigned requests, including new and critical service requests.
Each such not-assigned request triggers a run of SFS at the leaf datacenter co-located with the corresponding PoA.
\bu's goal is to find an initial feasible solution. 

Each datacenter runs this stage 
as follows. Let $s$ be the considered datacenter.
First, $s$ sorts the not-assigned requests 
 by 
an increasing 
number of higher-level datacenters that can host each request. Intuitively, a request with fewer degrees of freedom to be placed above $s$ must 
be placed for first.
Thus, $s$ tries to {\em assign} the requests according to the above order.
Assigning request $r$ means reserving a sufficient amount of computing resources (namely, CPU\footnote{Although we focus specifically on CPU as a computational resource, memory, and storage could be considered as well.}) 
to place it locally. 
If $r$'s latency constraint $d^{\max}_{\mv^r}$  allows placing $r$ above $s$, then
$s$ requests its ancestors to {\em push-up} $r$, to reduce costs. 
This {push-up} mechanism will participate in stage 2 of the algorithm, as detailed below.
If $r$'s latency constraint does not allow placing $r$ above $s$, then
$s$
places $r$ locally.  
If a request $r$ travels through Stage 1 all the way from its PoA up to its top latency-feasible datacenter $\bar{s}_r$ without being successfully assigned, then it is necessary to make room for $r$ to obtain a feasible solution. To do that, $s$ triggers Stage 3 ({\em Push-Down}), as described later.

\textbf{Stage 2: Push-Up (PU).} 
This algorithmic stage handles the push-up requests mentioned above. The goal now is to decrease the CPU cost by pushing up services. To prioritize pushing up services as high as possible, this algorithmic stage runs top-down. That is, a push-up request $r$ is considered first at $r$'s highest delay-feasible datacenter, $\bar{s}_r$; if $\bar{s}_r$ does not have enough available resources to host $r$, then $\bar{s}_r$ propagates the push-up request to its relevant child (the child on the path from $\bar{s}_r$ to the $r$'s PoA). A parent datacenter notifies its descendants about the status of their push-up requests using pos-ack (meaning that the service was pushed-up), or neg-ack (meaning that the service wasn't pushed-up).
If a datacenter $s$ receives a pos-ack for a service $r$ that $s$ previously assigned, then $s$ de-assigns $r$.

\textbf{Stage 3: Push-Down (PD).} 
If Stage 1 fails to assign a request $r$ on any of its latency-feasible datacenters $\Sc_r$, some services must be {\em pushed-down} to make room for $r$.
Hence, if $r$'s highest delay-feasible datacenter $\bar{s}_r$ cannot host $r$, then $\bar{s}_r$ initiates the push-down stage.
The initialization of the push-down procedure includes the following pieces of information: 
\begin{inparaenum}[(i)]
\item A list of push-down requests, and \item {\em deficitCPU}, namely, the amount of CPU resources that must be freed (i.e., pushed-down) from $\bar{s}_r$ to find a feasible solution. 

The push-down stage runs in $\bar{s}_r$'s sub-tree in a depth-first-search manner, until enough CPU is freed from $\bar{s}_r$. After freeing enough CPU, the push-down stage finds a feasible solution and finishes. 
If the push-down stage fails to free enough CPU from $\bar{s}_r$ even after scanning all $\bar{s}_r$'s sub-tree, the push-down procedure initiated by $\bar{s}_r$ terminates with a failure.
\end{inparaenum}

\subsection{Implementing \alg\label{sub:algo}}
In what follows, we describe how each stage is implemented through specific algorithmic procedures and define the protocol that allows the entities running such procedures to exchange the necessary information.

As our solution is distributed and asynchronous, several parallel simultaneous runs of each procedure -- namely, \bu(), \pu(), and \pd() -- may exist in the system.
To handle asynchronism, we prudently consider the different simultaneous runs of the algorithm on distinct nodes. 
In particular, we let $s$.proc() denote a run of procedure proc() on datacenter $s$. 

A procedure running on a certain datacenter may call another procedure in the same datacenter, e.g.,  $s$.proc1() may call $s$.proc2(). In addition, a procedure running on datacenter $s$ may call a procedure in {\em another} datacenter, $s'$, e.g.,  $s$.proc1() may call $s'$.proc1(), or even $s'$.proc2(). Such a call is manifested by sending a message from $s$ to $s'$. 

We detail and discuss the communication overhead in Sec.~\ref{sec:alg:assure_converge}. %

When multiple nodes run algorithms simultaneously, a resource competition occurs. We resolve this by prioritizing the requests as follows.
Observe that $S_r \setminus \Tc(s)$ represents the set of all higher-level datacenters on which request $r$ can be pushed up from $s$, while satisfying $r$'s delay constraints. 
Now $|S_r \setminus \Tc(s)|$ corresponds to the number of higher-level datacenters where $r$ can be placed: a smaller value corresponds to a more time-constrained request to place.
In the pseudo-codes of the algorithms, we use the procedure {\em Sort()} that sorts requests in increasing order of $\abs{S_r \setminus \Tc(s)}$, namely, from the most
time-constrained to the least one.
{\em Sort()} breaks ties by increasing $\beta^{r,s}$ and breaks further ties by new users last, in order to preserve already running services.

As our approach is distributed, each datacenter $s$ maintains its local variables, denoted by a sub-script $s$.  
In particular, we denote by $\notAssignedOfS$ the list of currently unassigned requests to be handled by datacenter $s$.
Further, we denote by $\pushUpSetOfS$ the list of requests assigned in 
$s$'s sub-tree, that 
may be pushed-up to an ancestor of $s$, to reduce costs.  
$\pushDownSetOfS$ denotes a list of push-down requests, with respect to $s$. Finally,  
$a_s$ denotes the available capacity on datacenter $s$, i.e., the capacity in $s$ that is not yet deployed or assigned to any service.

Upon system initialization, each datacenter $s$ resets
its local above three lists: 
$\notAssignedOfS{=}\pushUpSetOfS{=}\pushDownSetOfS{=}\emptyset$. In addition, each datacenter initially assigns the amount of available CPU to its computational capacity, namely, $a_s{=}C_s$.
Upon initialization, each datacenter $s$ resets a flag named within\_PD, which will be set if $s$ is currently running a \pd\ (push-down) procedure. 
If a service is migrated, the datacenter that hosted the migrating service automatically de-assigns (namely, de-allocates) the resources previously used by the migrating service.

In the following, we present how each stage is implemented.

\textbf{SFS stage implementation.} The SFS procedure is detailed in 
Alg.~\ref{alg:bu}. 
The procedure gets as input two lists: $\notAssigned$ -- a list of requests that are not assigned yet;
and $\pushUpSet$ -- a list of requests that are assigned,
but may be pushed-up to a higher datacenter, to reduce costs. 
In Line~\ref{alg:bu:sort3},  
\bu() adds the input list of not-assigned requests $\notAssigned$ to the local list of not-assigned requests $\notAssignedOfS$, and sorts $\notAssignedOfS$. 
The procedure does not sort the push-up requests because we only collect push-up requests; all the collected requests will be sorted later, by the \pu() procedure.
In Line~\ref{alg:bu:init_PD_set}, $s$ initializes the list of push-down requests to an empty list.

\bu() considers the un-assigned services as follows (Lines~\ref{alg:bu:for_begin}--\ref{alg:bu:for_end}). If the locally available capacity, denoted $a_s$, suffices to locally place an unassigned service $\mv^r$ (Line~\ref{alg:bu:if_enough_avail_CPU}), $s$ reserves capacity for $\mv^r$ (Line~\ref{alg:bu:dec_a_s}). If $\mv^r$ cannot be placed higher in the tree (Line~\ref{alg:bu:if_must_place}), \bu() not only assigns $\mv^r$, but also locally places it (Line~\ref{alg:bu:place}). Otherwise (namely, if $s$ does not have enough resources for $\mv^r$, but $\mv^r$ may be placed on an ancestor), the procedure inserts $\mv^r$ to the list of push-up requests   (Line~\ref{alg:bu:insert_to_push_up_list}). Later, $s$ will propagate these push-up requests to its parent. 
If \bu() fails to place a request $r$ that cannot be placed higher, $s$ adds $r$ to the push-down list (Lines~\ref{alg:bu:if_failed_begin}--\ref{alg:bu:if_failed_end}).
Note that the loop discards requests for which $\alpha_s<\beta^{r,s}$ and $\bar{s} \neq s$ (namely, $s$ does not have enough resources for $\mv^r$, but $\mv^r$ may be placed on an ancestor). Such requests will stay in $\notAssignedOfS$, and thus later be propagated to $s$'s parent (Line~\ref{alg:bu:call_prnt}).
After considering all the unassigned services, $s$ considers  push-down requests, if any  (Lines~\ref{alg:bu:if_pd_set_not_empty_begin}-\ref{alg:bu:if_pd_set_not_empty_end}). 
If they exist, then $s$ calls the push-down procedure as follows. First, $s$ calculates the amount of CPU that should be freed to find a feasible solution. This amount, denoted by \deficitCpu, is calculated as the gap between the amount of CPU required to place all the push-down requests, and the currently available CPU, $a_s$ (Line~\ref{alg:bu:set_defCpu}). Then $s$ initiates the \pd() procedure, with $\pushDownSetOfS$ and \deficitCpu\ as arguments (Line~\ref{alg:bu:call_pd}). 
If $s$ has no pending push-down requests (Line~\ref{alg:bu:if_pd_set_not_empty_else}),  
\bu() assigns the list of push-up requests to be sent to $s_{\textrm{prnt}}$, the parent of $s$. 

This list includes all the push-up requests that are (i) known to $s$, and (ii) can be deployed on $s_{\textrm{prnt}}$ while satisfying the latency constraint (Line~\ref{alg:bu:set_puList_to_prnt}).
Later, \bu() checks whether there exist services that are not yet assigned, or can be pushed-up to an ancestor (Line~\ref{alg:bu:if_should_call_prnt}).

If so, \bu() initiates a run of \bu() on $s$'s parent (Line~\ref{alg:bu:call_prnt}). 
Finally, if there are no pending push-up requests from its ancestors, $s$ initiates a PU procedure (Lines~\ref{alg:bu:if_no_pnding_pu_from_prnt}--\ref{alg:bu:init_pu}).

\begin{algorithm}[tb]
								 
\caption{\bu ($ \notAssigned, \pushUpSet$)}
 \algFontSize
\label{alg:bu}
\algFontSize
\begin{algorithmic}[1]
\State {$s \gets $ local datacenter id}
\label{alg:bu:set_s}
\State $\pushUpSetOfS \leftarrow \pushUpSetOfS \cup \pushUpSet$
\label{alg:bu:init_push_up_set}
\State $\notAssignedOfS \leftarrow$ Sort ($\notAssignedOfS \cup \notAssigned$) 
\label{alg:bu:sort3}
\State $\pushDownSetOfS = \emptyset$ 
\label{alg:bu:init_PD_set}
\ForEach {$\mv^r \in \notAssignedOfS$}
    \label{alg:bu:for_begin}
    \If {$a_s \geq \beta^{r,s}$}
    \Comment{enough available CPU to place $\mv^r$ on $s$}
    \label{alg:bu:if_enough_avail_CPU}
    \State remove $\mv^r$ from $\notAssignedOfS$
    \label{alg:bu:pup_from_Hcomp}
    \State $a_s \leftarrow a_s - \beta^{r,s}$        
    \Comment{assign $\mv^r$ on $s$}
    \label{alg:bu:dec_a_s}
    \If {$\bar{s}_r = s$} 
    \Comment{check if $\mv^r$ cannot be placed higher}
      \label{alg:bu:if_must_place}
        \State place $\mv^r$ on $s$
        \label{alg:bu:place}
    \Else
    \Comment{Will place $\mv^r$ on $s$ only if it won't be pushed-up}
        \State         $\pushUpSetOfS \leftarrow  \pushUpSetOfS\cup \set{r}$\Comment{Insert $r$ to $ \pushUpSetOfS$}
        \label{alg:bu:insert_to_push_up_list}
    \EndIf
    \ElsIf {$\bar{s}_r = s$} 
        \label{alg:bu:if_failed_begin}
        \Comment{if $\mv^r$ can't be placed neither here nor on an ancestor}
        \State $\pushDownSetOfS \leftarrow  \pushDownSetOfS\cup \set{r}$\Comment{Insert $r$ to $ \pushDownSetOfS$}
        \label{alg:bu:insert_to_pd_list}
    \EndIf \label{alg:bu:if_failed_end}
\EndFor \label{alg:bu:for_end}
   \If{$\pushDownSetOfS\neq\emptyset$}
   \label{alg:bu:if_pd_set_not_empty_begin}
        \State $\deficitCpuMath = \sum_{r \in \pushDownSetOfS} \beta^{r,s} - a_s$
        \label{alg:bu:set_defCpu}
        \Comment{capacity to free for finding a sol}
        \State run $\pdMath\ (\pushDownSetOfS, \deficitCpuMath)$ 
        \label{alg:bu:call_pd}
                    \Comment{Push-down services to make room for over-provisioned services}
       \label{alg:bu:if_pd_set_not_empty_end}
\Else \label{alg:bu:if_pd_set_not_empty_else}
    \State $\pushUpSet_{\textrm{prnt}} \leftarrow$ all requests in $\pushUpSetOfS$ for which $s_{\textrm{prnt}}$ is delay-feasible
    \label{alg:bu:set_puList_to_prnt}
    \If {$\notAssignedOfS \neq \emptyset$ \OR $\pushUpSet_{\textrm{prnt}} \neq \emptyset$}
        \Comment{check if there is any req to send to parent}
        \label{alg:bu:if_should_call_prnt}
        \State{send \Big($s$.\textrm{prnt}, \bu \big($\notAssignedOfS, \pushUpSet_{\textrm{prnt}}$\big)\Big)}
        \label{alg:bu:call_prnt}
        \State $\pushUpSet_s = \pushUpSet_s \setminus \pushUpSet_{\textrm{prnt}}$
    \EndIf        
    \If {$\pushUpSet_{\textrm{prnt}} = \emptyset$} 
        \Comment{Check if there are no pending PU replies from parent}
        \label{alg:bu:if_no_pnding_pu_from_prnt}
        \State run \pu $\left(\pushUpSetOfS \right)$
        \label{alg:bu:init_pu}
    \EndIf
\EndIf
\end{algorithmic}
\end{algorithm}

\begin{algorithm}[tb!]
\caption{\pu ($\pushUpSetOfS$)}
 \algFontSize
\label{alg:pu}
\begin{algorithmic}[1]
\State {$s \gets $ local datacenter id}
    \State de-assign all services previously pushed-up from $s$, and update $a_s$ and $\pushUpSetOfS$ accordingly
    \label{alg:pu:displace}
    \State $\pushUpSetOfS \leftarrow$ Sort ($\pushUpSetOfS$)\Comment{Sort by latency 
    constraints
    }%
    \label{alg:pu:sort}
    \ForEach {push-up request $r$ in $\pushUpSetOfS$}
        \label{alg:pu:for_pu_begin}
        \If {$a_s \geq \beta^{r,s}$}
        \label{alg:pu:enough_as}
        \Comment{if enough CPU to place this service}
        \State {place $\mv^r$ on $s$ and update $a_s$ accordingly}
        \label{alg:pu:pu}

        \label{alg:pu:inform}
        \EndIf
        \State set a pos-ack for $r$ if it was pushed-up, neg-ack else
    \EndFor    
    \label{alg:pu:for_pu_end}
   
    \ForEach{child $c$}\label{alg:pu:for_each_child_begin}
        \State $\pushUpSet_c \gets$ all push-up requests originated from  $c$ subtree
        \label{alg:pu:associate_puReq_to_child}
        \If{$\pushUpSet_c \neq \emptyset$}
            \label{alg:pu:forEach_child_if_begin}
            \Comment{if there are requests pushed-up from $c$'s sub-tree}
            \State send \Big($c$, \pu \big($\pushUpSet_c\big)\Big)$
        \EndIf
        \label{alg:pu:forEach_child_if_end}
    \EndFor\label{alg:pu:for_each_child_end}
\end{algorithmic}
\end{algorithm}

\begin{algorithm}[tb!] 					   
\caption{Process PD message($\pushDownSetOfS$,$\deficitCpuMath$) }
\algFontSize
\label{alg:pd}
\begin{algorithmic}[1]
\State{$s\gets$ local 	datacenter id}  
\Comment{$s$ called by \pd ($\pushDownSetOfS, \deficitCpuMath)$}
\label{alg:pd:read_fields_begin}
    \If {within\_PD}\Comment{Note that flag is initially set false}
        \label{alg:pd:if_within_another_begin}
        \State send the caller neg-ack for all the requests in $\pushDownSetOfS$
        \State \Return
        \Comment{cannot handle more than one \pd() request simultaneously}
    \label{alg:pd:if_within_another_end}
    \EndIf
    \State within\_PD = True
    \State add all the requests assigned but not placed on $s$ to the end of
    $\pushDownSetOfS$
    \label{alg:pd:add_pot_placed}
    \State add all the requests placed on $s$ to the end of     $\pushDownSetOfS$
    \label{alg:pd:add_placed}
    \Comment{Prioritize sols with fewer mig.}
    \ForEach {child $c$} 
        \label{alg:pd:foreach_child_begin}
        \If {can place on $s$ enough services from $\pushDownSetOfS$ to make deficitCpu $\leq 0$}
        \label{alg:pd:if_can_break_next_child_begin}
            \State break \label{alg:pd:if_can_break_next_child}
        \EndIf
        \label{alg:pd:if_can_break_next_child_end}
        \State $\pushDownSet_c \gets$ requests in $\pushUpSetOfS$ originated from  $c$  
        \If {$\pushDownSet_c \neq \emptyset$}
            \label{alg:pd:call_child_begin}
            \Comment{if child $c$ may help in freeing space}
            \State send \big($c, \pdMath$  ($\pushDownSet_c, \deficitCpuMath$)\big)
            \Comment{request $c$ and get ack}
            \State de-assign services pushed-down from $s$; update $a_s$,  $\pushDownSetOfS, \notAssignedOfS$, \deficitCpu
            \label{alg:pd:update_by_pushed_dwn_chains_from_me}
        \EndIf
        \label{alg:pd:call_child_end}
    \EndFor
    \label{alg:pd:foreach_child_end}
    \ForEach {$r \in \pushDownSetOfS$ s.t. $\beta^{r,s} \leq a_s$}
    \label{alg:pd:pd_to_myself_begin}
        \State {place $\mv^r$ and update $a_s$, $\pushDownSetOfS$ and deficitCpu accordingly}
        \label{alg:pd:pd_to_myself_end}
    \EndFor
    \If {$s$ is not the initiator of this \pd} 
    \label{alg:pd:if_need_callPrnt}
        \State send \big($s$.parent ($\pushDownSetOfS, \deficitCpuMath$)\big)
        \label{alg:pd:callPrnt}
    \EndIf
    \State run $s$.\bu($\notAssignedOfS$) in $F$-mode\Comment{Run SFS on un-assigned requests}
    \label{alg:pd:run_f_sfs}
    \State within\_PD = False
\end{algorithmic}
\end{algorithm}

\textbf{PU stage implementation.} The push-up procedure \pu() is detailed in Alg.~\ref{alg:pu}. The procedure first {deassigns}
 and regains the CPU resources for all the services pushed-up from $s$ to a higher-level datacenter (Line~\ref{alg:pu:displace}).
In Lines~\ref{alg:pu:for_pu_begin}--\ref{alg:pu:for_pu_end}, \pu() considers all the push-up requests as follows. Consider a request to push up service $\mv^r$, currently placed on datacenter $s'$, a descendent of $s$. If $s$ has enough available capacity for that request, then \pu() locally places $\mv^r$ and updates the relevant record in $\pushUpSetOfS$ (Lines~\ref{alg:pu:enough_as}--\ref{alg:pu:pu}). This record will later be propagated to $s'$, which will identify that $\mv^r$ was pushed up, and regain the resources allocated for it. 
Lines~\ref{alg:pu:for_each_child_begin}--\ref{alg:pu:for_each_child_end} state that $s$ propagates the push-up requests to its children. To reduce communication overhead, the procedure associates each record in $\pushUpSetOfS$ with a single child $c$ that is latency-feasible for the service in question (Line~\ref{alg:pu:associate_puReq_to_child}), and propagates to each child only the relevant requests, if any exist (Lines~\ref{alg:pu:forEach_child_if_begin}--\ref{alg:pu:forEach_child_if_end}).

\textbf{PD stage implementation.} 
The push-down procedure, \pd(), runs the same when either a parent calls its child, or vice versa.  
\pd() runs in a depth-first-search manner in the sub-tree rooted by the datacenter that initiated it. As \alg\ is asynchronous, several instances of \pd() may run simultaneously. 
To assure convergence and reduce the algorithm's overhead, we let each datacenter $s$ take part in, at most, a single instance of \pd() at any time. This is enforced by introducing a flag named {\em within\_PD}, which is set when a \pd() is active. 
If \pd() calls $s$ when within\_PD is set, $s$ replies with a negative acknowledgment for all the push-down requests. 

Recall that \deficitCpu\ denotes the amount of resources that should be pushed-down from the datacenter that initiated the \pd() procedure to find a feasible solution.
Each time a service is pushed-down from the datacenter that initiated \pd(), \deficitCpu\ is decremented accordingly. A run of \pd() is terminated once \deficitCpu\ is not positive anymore, indicating that enough resources were freed. Due to the granularity of the services' CPU consumption, 
when \pd() terminates, \deficitCpu\ may be either zero or negative, namely, \pd() terminates when \deficitCpu$\leq 0$. 

The processing of the messages in \pd() is detailed in Alg.~\ref{alg:pd}.
Lines~\ref{alg:pd:if_within_another_begin}--\ref{alg:pd:if_within_another_end} handle the particular case where \pd() is invoked while $s$ already takes part in another run of \pd(). In such a case, the procedure merely sends neg-acks for all the push-down requests and returns.
In Lines~\ref{alg:pd:add_pot_placed}--\ref{alg:pd:add_placed}, \pd() 
adds to the list of requests $\pushDownSetOfS$ the locally assigned services. To reduce the number of migrations, $s$ adds the services that are not merely assigned but already really placed to the end of $\pushDownSetOfS$
(Line~\ref{alg:pd:add_placed}). Thus, the procedure will consider migrating services only if it does not find a solution without migrating already-placed services. 

In Lines~\ref{alg:pd:foreach_child_begin}--\ref{alg:pd:foreach_child_end}, $s$ serially requests its children to push-down services. 
Before each such call to a child, $s$ checks whether making \deficitCpu\ $\leq 0$ is possible without calling an additional child. If the answer is positive, $s$ skips calling its children (Lines~\ref{alg:pd:if_can_break_next_child_begin}--\ref{alg:pd:if_can_break_next_child_end}). Upon receiving a reply from child $c$, $s$ updates its state variables and \deficitCpu\ accordingly (Lines~\ref{alg:pd:call_child_begin}--\ref{alg:pd:call_child_end}).
In Lines~\ref{alg:pd:pd_to_myself_begin}--\ref{alg:pd:pd_to_myself_end}, $s$ tries to locally host services from the push-down list, $\pushDownSetOfS$. 
Later, if $s$ is not the initiator of this push-down procedure, it replies to its parent (Lines~\ref{alg:pd:if_need_callPrnt}--\ref{alg:pd:callPrnt}). 
Finally, if $s$ has services that are not assigned yet, 
\pd() calls \buFMode(), a degenerated version of \bu() that focuses on finding a feasible solution; we detail  \buFMode() in Sec.~\ref{sec:alg:assure_converge}.

\begin{figure*}
    \centering
    \includegraphics[width=2\columnwidth]{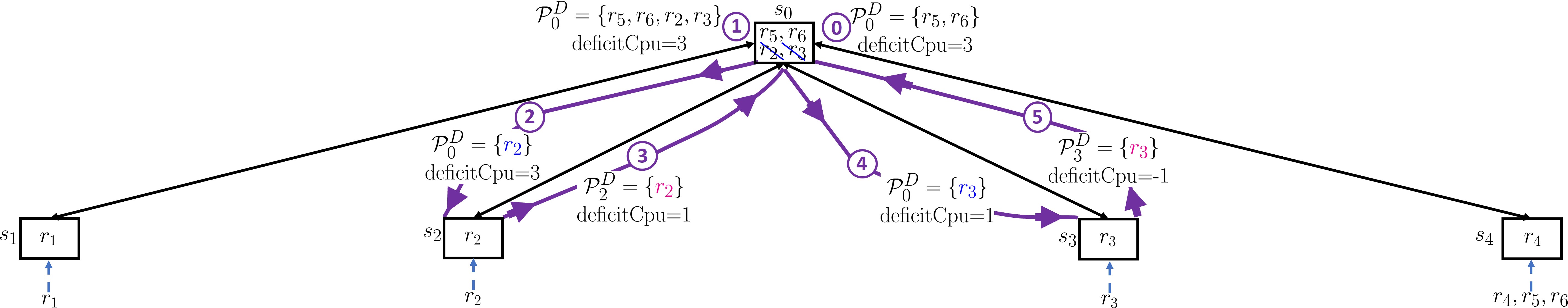}
        \caption{Run example of the \pd\ procedure. The circle denotes the step $\in\{0,5\}$.}
        \label{fig:pd_toy_example}
\end{figure*}

\example{}{Fig.~\ref{fig:pd_toy_example} exemplifies a run of \pd(). 
Each request $r$ can be placed either on its PoA or on the root $s_0$, with a fixed required capacity $\beta^{r,s}{=}2$.
The normalized CPU capacities of the datacenters are $C_0{=}5, C_1{=}C_2{=}C_3{=}C_4{=}2$. 
At the beginning (Step 0), service $\mv^{r_1}$ is placed on $s_1$; services $\mv^{r_2}$ and $\mv^{r_3}$ are placed on $s_0$; and service $\mv^{r_4}$ is placed on $s_4$. 
When the new requests $r_5, r_6$ arrive, their latency-feasible datacenters (namely, $s_4$ and $s_0$) do not have sufficient available capacity. 
In particular, when $s_0$.\bu() runs with $\notAssigned_0 {=} \set{r_5, r_6}$, there is not enough capacity for $r_5$. As there is no ancestor of $s_0$ on which $r_5$ can be placed, $s_0.$\bu() inserts $r_5$ into $\pushDownSetOfS$ (Line~\ref{alg:bu:insert_to_pd_list} in Alg.~\ref{alg:bu}). When reaching Line~\ref{alg:bu:insert_to_pd_list} in the next iteration of the for loop, $s$ will fail to assign also $r_6$ and, thus, insert $r_6$ to the push-down list.
Next, $s$ sets \deficitCpu\ {=} $\sum_{r {\in} \pushDownSet_{s_0}} \beta^{r,s} {-} a_{s_0} {=} 4-1 {=} 3 $ (ln.~\ref{alg:bu:set_defCpu} in Alg.~\ref{alg:bu}); and then calls  the \pd() procedure (Line~\ref{alg:bu:call_pd} in Alg.~\ref{alg:bu}). 

Step 1 in Fig.~\ref{fig:pd_toy_example} displays the run of $s_0$.\pd(). 
The procedure adds to $\pushDownSet_0$ all the services that are currently assigned to $s_0$, namely, $r_2$ and $r_3$ (Lines~\ref{alg:pd:add_pot_placed}--\ref{alg:pd:add_placed} in Alg.~\ref{alg:pd}). Pushing-down such services to descendants of $s_0$ may free capacity, that $s_0$ may use later to place $r_5$ and/or $r_6$.
Then (Lines~\ref{alg:pd:foreach_child_begin}--\ref{alg:pd:foreach_child_end}), $s_0$ considers calling its children. The child $s_1$ is not a latency-feasible datacenter for any of the services in $\pushDownSet_0$; hence, $s_0$.\pd() does not call it. 

Step 2 in Fig.~\ref{fig:pd_toy_example} captures how  $s_0$.\pd() calls $s_2$.\pd() with $\pushDownSet_0 {=}\set{r_2}$ and \deficitCpu{=}3. As a result, $s_2$.\pd() pushes-down $r_2$, and updates \deficitCpu\ accordingly (Line~\ref{alg:pd:pd_to_myself_end} in Alg.~\ref{alg:pd}). That is, $s_2$.\pd() decreases \deficitCpu\ by the amount of CPU freed from $s_0$ by pushing-down $s_2$: \deficitCpu\ ${=} 3 {-} 2 {=} 1$. Then, $s_2$.\pd() sends a push-down Ack message to its parent, $s_0$ (Line~\ref{alg:pd:callPrnt} in Alg.~\ref{alg:pd}). This call is depicted in Step~3 in Fig.~\ref{fig:pd_toy_example}.

In Step 4, $s_0$.\pd() calls $s_3$ with $\pushDownSet_0 {=}\set{r_3}$ and \deficitCpu=1. 
Consequently, $s_3$.\pd() pushes-down $r_3$, and updates \deficitCpu\ again: \deficitCpu\ ${=}1 {-} 2 {=} {-}1$. 
The negative value assigned to \deficitCpu\ reflects the fact that not only  pushing-down services $r_2$ and $r_3$ to from $s_0$ frees enough CPU capacity in $s_0$,
but also after placing $r_5$ and $r_6$ (the ``problematic'' requests that initiated the \pd() procedure  the first time), $s_0$ will still have 1 unit of capacity available.

In Step 5, $s_3$.\pd() sends an ack message to its parent, $s_0$. 
As $s_0$.\pd() realizes that  \deficitCpu\ ${\leq} 0$, it breaks the loop (Lines~\ref{alg:pd:if_can_break_next_child_begin}--\ref{alg:pd:if_can_break_next_child_end} in Alg.~\ref{alg:pd}). 
Next, $s_0$ ``pushes-down'' to itself the problematic services $r_5$ and 
$r_6$ (Lines~\ref{alg:pd:pd_to_myself_begin}--\ref{alg:pd:pd_to_myself_end}  in Alg.~\ref{alg:pd}), and calls $F$-\bu() (Line~\ref{alg:pd:run_f_sfs}  in Alg.~\ref{alg:pd}). 
This completes the run of the push-down procedure.}

\subsection{Assuring convergence}\label{sec:alg:assure_converge}

Intuitively, a run of \pd() indicates that a recent run of \bu() -- either in $s$, or in an ancestor of $s$ -- failed, and hence \pd() is called. In such circumstances, pushing-up services again may result in a non-convergence, manifested as endless oscillations of pushing-up and pushing-down the same services. 
As an illustration of such possible oscillations, consider again the scenario depicted in Fig.~\ref{fig:need_bu_problem}. To successfully place the new request $r_3$, the algorithm would push-down $r_1$ from $s_2$ to $s_5$, thus reaching the state depicted in Fig.~\ref{fig:need_bu_sol}. Now, assume that the service request $r_3$ finishes, thus leaving $s_2$ totally free, as depicted in Fig.~\ref{fig:need_bu_r_3_finish}. To reduce the cost, the algorithm may push-up $r_1$ back to $s_2$, as depicted in Fig.~\ref{fig:need_bu_after_pu}. Now, when a new service request, $r_4$, arrives at PoA $s_6$, the algorithm reaches the state depicted in Fig.~\ref{fig:need_bu_fail_again}, which is actually isomorphic to the state in Fig.~\ref{fig:need_bu_problem}; in particular, to place $r_4$, the algorithm has to push-down $r_1$ again, back to $s_5$. 
Such oscillations cause multiple unnecessary migrations and communication overhead.

To avoid such oscillations, our algorithm deters new push-ups of services shortly after pushing-down services. 
In more detail, we define a Feasibility {\em F}-protocol mode. Each time $s$ calls \pd(), $s$ enters $F$-mode (if it was not already in $F$-mode), and remains so for some pre-configured {\em F-mode period}.
While in $F$-mode, \alg\ does not initiate new push-up requests. However, while in $F$-mode, the algorithm still replies to existing push-up requests with the minimum necessary details that the caller node is waiting to receive, to prevent deadlocks.

Alg.~\ref{alg:buFMode} details \buFMode(), which is the $F$-mode equivalent of \bu(). The procedure is a degenerated version of \bu() that places every service and does not initiate new push-up requests. 
If buFMode() fails to place a request even after calling \pd() (Line~\ref{alg:buFMode:if_failure}), calling \pd() again may result in an infinite loop of trying to push down services and seek a feasible solution while such a solution does not exist. Hence, instead of calling \pd() again, the procedure terminates \alg\ with a failure (Line~\ref{alg:buFMode:failure}).\footnote{This failure can be avoided if enough resources are made available.} 

\buFMode() does not initiate new push-up requests. However, to prevent deadlocks, the procedure considers existing push-up requests. This is done by calling a degenerated version of the push-up procedure, named F-\pu()  (Line~\ref{alg:buFMode:call_puFmode}). $F$-\pu() is identical to \pu(), except for the following difference: instead of handling the push-up requests (Line~\ref{alg:pu:for_pu_begin}--\ref{alg:pu:for_pu_end}), $F$-\pu() merely sets neg-acks for all the push-up requests.

We now show that once there are no new requests, the algorithm converges in a finite time.
\begin{theorem}\label{thm:converge}
Let $\Rc^*$ denote the set of requests that are not placed yet at some point in time. If there are no new requests, then after exchanging $O \left(\abs{\Rc^*} \cdot \abs{\Sc} \right)$ messages, the protocol either fails or finds a feasible solution.
\end{theorem}
\begin{proof}
In the worst case, each request $r \in \Rc^*$ initiates a distinct run of \alg, and each such run involves running \bu(), \pu(), and \pd().
A single run of \bu() or \pu() exchanges $O(D(\Gc))$ messages, where $D(\Gc)$ is the diameter of the network graph $\Gc$ -- from the PoA and upwards when running \bu(), or vice versa when running \pu(). 

The initiator of the \pd\ procedure is $r$'s top latency-feasible datacenter, $\bar{s}_r$ (Lines~\ref{alg:bu:if_failed_begin}--\ref{alg:bu:if_failed_end} in Alg.~\ref{alg:bu}). 
In the worst-case, $\bar{s}_r$ is the root datacenter, and the run of \pd\ involves sending messages down to all the leaves datacenters and back up to the root -- namely, $O \left(\abs{\Sc}\right)$ messages. 
After this run, the algorithm either fails or successfully places at least one additional request. To conclude, the protocol either fails or finds a feasible solution after sending $O\left(\abs{\Rc^*} \cdot (D(\Gc) {+} \abs{\Sc})\right) {=} O\left(\abs{\Rc^*} \cdot \abs{\Sc}\right)$.
\end{proof}

Although the upper bound on the number of messages required for convergence in  Theorem~\ref{thm:converge} is loose, in Sec.~\ref{sec:reduce_commoh} we show how to dramatically decrease the overhead of \alg\ in practical cases.

\begin{algorithm}[tb!]
\caption{\buFMode ($\notAssigned, \pushUpSet$)}
\algFontSize
\label{alg:buFMode}
\begin{algorithmic}[1]
    \State run Lines~\ref{alg:bu:set_s}-\ref{alg:bu:init_PD_set} of Alg.~\ref{alg:bu}
    \ForEach {$\mv^r \in \notAssignedOfS$}
        \State remove $\mv^r$ from $\notAssignedOfS$
        \If {$a_s \geq \beta^{r,s}$}
        \Comment{if there is enough available CPU to place $\mv^r$ on $s$}
            \State place $\mv^r$ and decrease $a_s$ accordingly
        \Else
            \If{$\bar{s}_r = s$} 
                \label{alg:buFMode:if_cannot_place_higher}
                \Comment{if cannot place this request on a higher datacenter}
                \If{within\_\pd==False} \Comment {if didn't try yet to push-down}
                    \State $\pushDownSetOfS \leftarrow \pushDownSetOfS \cup \set{r}$
                    \label{alg:buFMode:set_pdList}
                    \Comment{services to push-down}
                \Else \label{alg:buFMode:if_failure}
                    \State \Return FAILURE  
                    \label{alg:buFMode:failure}
                \EndIf
            \EndIf
            \label{alg:buFMode:if_cannot_place_higher_endif}
        \EndIf
    \EndFor
   \If{$\pushDownSetOfS\neq\emptyset$}
   \label{alg:buFMode:if_pd_set_not_empty_begin}
        \State $\deficitCpuMath = \sum_{r \in \pushDownSetOfS} \beta^{r,s} - a_s$
        \label{alg:buFMode:set_defCpu}
        \Comment{capacity to free for finding a sol}
        \State run $\pdMath\ (\pushDownSetOfS, \deficitCpuMath)$ 
        \label{alg:buFMode:call_pd}
                    \Comment{Push-down services}
       \label{alg:buFMode:if_pd_set_not_empty_end}    
    \EndIf
    \If {$\notAssignedOfS \neq \emptyset$} 
        \State{send \Big($s$.parent, \bu  \big($\notAssignedOfS, \emptyset$\big)\Big)}
    \EndIf
    \State {\bf run} $s$.F-\pu($\pushUpSetOfS$)
\label{alg:buFMode:call_puFmode}
\end{algorithmic}
\end{algorithm}

\subsection{Optimizing the implementation}\label{sec:reduce_commoh} 
The upper bound implied by Theorem~\ref{thm:converge} assumes that each attempt to place a single service triggers a unique run of the protocol, thus incurring excessive overhead. 
To avoid such a case, observe that, typically, several users move together in the same direction (e.g.,  cars moving simultaneously on the road, on the same trajectory). 
Naively, such a scenario may translate into multiple invocations of \alg, each of them for placing a single requested service instance. 
Furthermore, initiating
too many simultaneous runs of \bu() that compete on the restricted resources in the same sub-tree 
may result in unnecessary protocol failures. 
To tackle this problem, we 
introduce short {\em accumulation delays} to our protocol. 
We let each datacenter receiving an \bu\ message wait for a short {\em \bu\ accumulation delay} before handling the new request. 
To deter long service latency or even deadlocks, each datacenter maintains a local \bu\ accumulation delay timer that operates as follows: if a run of \bu() on $s$ reaches Line~\ref{alg:bu:for_begin} in Alg.~\ref{alg:bu} while no \bu\ accumulation delay timer is ticking on $s$, the procedure initiates a new \bu\ accumulation delay timer. 
This current run of \bu and all the subsequent runs of \bu() in $s$ halt. After the \bu\ accumulation delay terminates, only a single \bu() process resumes in $s$.

Likewise, to initiate fewer runs of \pd(), we let each datacenter retain a single \pd\ accumulation delay mechanism that works similarly to the \bu\ accumulation delay mechanism. 
Significantly, the accumulation delay only impacts the time until the protocol finds a new feasible placement, not the latency affecting the data-plane information transfer. We assess the impact of the accumulation delay in Sec.~\ref{sec:sim:results}. 

\section{Performance Evaluation}\label{sec:sim} 
In this section, we compare the performance of \alg\ to state-of-the-art solutions. The comparison provides insights into the impact of the various system parameters.

\subsection{Experiments settings}\label{sec:sim_settings}

\textbf{Service areas and user mobility.}
We consider two real-world mobility traces, representing scenarios with distinct characteristics, that is, the vehicular traffic within the centers of the cities of (i) Luxembourg~\cite{Luxembourg},  and (ii) the  Principality of Monaco~\cite{Monaco}.   
For the PoAs, we rely on the real-world antenna locations publicly available in~\cite{OpenCellid}. For each of the above geographical areas, we consider the antennas of the cellular telecom provider having the largest number of antennas in the area.
For both traces, we consider the 8:20–8:30\,AM rush hour period. 
The settings of the service areas are detailed in Table~\ref{tab:sim_settings}. 
Further details about the mobility traces can be found in~\cite{Dynamic_Service_Provisioning_ToN, Luxembourg, Monaco}.

\begin{table}
    \scriptsize
    \centering
    \caption{Reference scenarios}
    {    \begin{tabular}{l|c|c|c|c|c|}
        \cline{2-6}
        & Area   & \#PoAs & \#Distinct & Avg.\ density       & Avg.\ speed \\
        &  [km$^2$]&  & vehicles   & [\#req/km$^2$] & [km/h] \\
        \hline
        \multicolumn{1}{|l|}{Luxemburg} & $6.8 {\times} 5.7$ & 1,524 & 25,497 & 56   & 15.4 \\
        \hline
        \multicolumn{1}{|l|}{Monaco}     & $3.1 {\times} 1$   & 231  & 13,788 & 2,297 & 9.0 \\
        \hline
    \end{tabular}
    }
    \label{tab:sim_settings}
\end{table}
		  
\textbf{Network and datacenters.}
The edge-cloud multi-tier network is structured as a 6-height tree;  a topology level is denoted by $\ell{\in}\{0, 1, \dots, 5\}$. The leaves (level 0) are the datacenters at the PoAs (i.e., co-located with the radio points of access). 
Similarly to~\cite{Avatar, dynamic_sched_and_reconf_t, NFV_ego_learning}, the higher levels $5, 4, 3, 2, 1$ recursively partition the simulated area. 
In both Luxembourg and Monaco, if no PoAs exist in a particular rectangle, the respective datacenters are pruned from the tree. 
The CPU capacity increases with the level $\ell$ to reflect the larger computational capacity in datacenters closer to the cloud. Denoting the CPU capacity at each leaf datacenter by $C_\text{cpu}$, the CPU capacity at level $\ell$ is $(\ell+1) \cdot C_\text{cpu}$. 
	 
\textbf{Services and costs.} Each vehicle that enters a considered geographical area is randomly marked as requesting either real-time (RT) or non-RT services, with some probability defined later. 
We set the RT latency constraint
 equal to 10~ms, corresponding to an automotive safety service (e.g., collision avoidance~\cite{ca})  with tight delay constraints,
while the non-RT latency constraint is 100~ms, representing, e.g., a see-through application~\cite{see-through}. 

We set our objective function's migration and computation cost functions in~\eqref{Eq:def_obj_func} such that they are no more than an order of magnitude apart, and none becomes negligible under our scenario's specific settings. 
Accordingly, the computation cost of 100~MHz of CPU at level $\ell$ is $\chi^c= 2^{5-\ell}$ cost units, to reflect the decrease in computation costs when moving from the edge to the cloud~\cite{tong2016_justify_path_to_root, SFC_mig_short}. 

To account for the communication cost, we add a cost of 3 for each bit traversing a link. The migration cost is $\chi^m=600$ cost units. For further discussion of these cost parameters, see~\cite{Dynamic_Service_Provisioning_ToN}. 

We calculate the required resources $\Sc_r$, $\beta^{r,s}$, and the implied costs $\psi^c(r,s)$ for each $r {\in} \Rc$ and $s {\in} \Sc$ using the 
\cpuall\ (GetFeasibleAllocations) algorithm~\cite{Dynamic_Service_Provisioning_ToN}
We thus obtain the following values.   
Each RT request can be placed at levels 0, 1, or 2 in the tree, requiring CPU of 17, 17, and 19~GHz, associated with costs of 544, 278, and 164, respectively. 	 
Each non-RT request can instead be placed at any level, with a fixed allocated CPU of 17~GHz and associated costs of 544, 278, 148, 86, 58, and 47 for placing the service at levels 0, 1, 2, 3, 4 and 5, respectively. 
These costs reflect the significantly lower CPU costs closer to the cloud. 
The migration cost is $\psi^m(r,s,s') {=} 600$ for every request $r$ and datacenters $s, s'$.
We assign the size of the control-plane messages according to Table~\ref{tab:pkt_fields_size}. 
In Table~\ref{tab:pkt_fields_size}, we consider that each service belongs to a specified {\em class} of latency constraint.
Given the classId of request $r$, each datacenter $s$ can calculate the amount of CPU required for placing service $\mv^r$ on $s$. The fields $\abs{\textrm{serviceId}}$ and $\abs{\textrm{dcId}}$ represent the number of bits required to represent the ID of a service and the ID of a datacenter, respectively.
$\abs{\beta^{u,s}}$ represents the number of bits used to specify the CPU resources a service $u$ needs when placed on datacenter $s$.
$\ceil*{\log_2 C_r}$ indicates the number of bits required to represent deficitCpu, namely, the amount of CPU to be pushed-down. In any reasonable system, deficitCpu is, at most, the capacity of the root datacenter $C_r$. 

\begin{table}
    \scriptsize
    \centering
    \caption{Size of packet's fields [bits]}
    {    
\begin{tabular}{|c|c|c|c|c|c|}
        \hline
        \rule{0pt}{2.4ex}    
        Header & $\abs{\textrm{dcId}}$ & $\abs{\textrm{serviceId}}$ & $\abs{\textrm{classId}}$ & 
        $\abs{\beta^{u,s}}$ &
        $\ceil*{\log_2 C_r}$ \\ \hline
        80 & 12 & 14 & 4 &5 & 16 \\ \hline
    \end{tabular}
    }
    \label{tab:pkt_fields_size}
\end{table}

\textbf{Delays.}
 To assess the effect of the network delays on the decision process of \alg, we focus on the control traffic exchanged between datacenters, considering both the transmission delay and the propagation delay at the packet level. 
The transmission delay is calculated based on the capacity allocated for the control plane at each link, through a dedicated network slice. We set the control plane capacity equal to 10~Mbps. 
For the propagation delay, we assume a constant propagation delay for each link, related to the diameter of the considered area.
Consequently, the propagation delay of each link in Luxembourg and Monaco is $22~\mu \text{s}$ and $8~\mu \text{s}$, respectively.  
For a datacenter at level $\ell$,  \bu\ accumulation delay, and \pd\ accumulation delay are, respectively,
$(\ell+1) \cdot \accDelayBu$, 
and $(\ell+1) \cdot \accDelayPd$. 
To make the accumulation delays an order of magnitude smaller than the propagation delays stated above, we set by default $\accDelayBu{=}0.1$\,ms. To decrease the number of push-down runs, we set the push-down accumulation delay to a slightly higher value, namely, $\accDelayPd{=}0.4$\,ms. These are default values, that we later vary to study their impact on the algorithm's performance.  

Each time the protocol calls the \pd() procedure, it stays in {\em F} mode (recall Sec.~\ref{sec:alg:assure_converge}) for the next 10 seconds.

\textbf{Benchmark algorithms.} 
To the best of our knowledge, no fully distributed, asynchronous algorithm exists for the \migProb\ we posed. Hence, we consider {\em centralized} placement schemes that identify the currently critical and newly requested services once every second, and solve the respective \migProb. Specifically, we consider the following algorithms.
							  
{\em Lower Bound (LBound):} An optimal solution to the relaxed version of \migProb, which can place fractions of a service on distinct datacenters.
It is computed using Gurobi optimizer~\cite{Gurobi}.
 Also, the LP formulation can consider and act upon not just critical services, but all services in the system every 1~s period. Hence, it serves as a lower bound on the cost of any feasible solution to the problem.

{\em \ffit}: It places each request $r$ on the lowest datacenter in $\Sc_r$ that has sufficient available resources to place service $\mv^r$. This is an adaptation to our problem of the placement algorithm proposed in Sec.\,IV.B in~\cite{App_placement_in_fog_n_edge}. 

{\em BUPU~\cite{Dynamic_Service_Provisioning_ToN}:} It consists of two stages. At the {\em bottom-up} stage, it places all the critical and new services as low as possible in the tree topology. If this stage fails to place a service while considering datacenter $s$, the algorithm re-places {\em all} the services associated with $s$'s sub-tree from scratch. 
Later, BUPU performs a push-up stage similar to our \alg's \pu() procedure.  In summary, BUPU can be considered a centralized version of \alg.

{\em \cpvnf~\cite{CPVNF_proactive_place_in_CDN}:} It orders the critical and newly-arriving services in non-increasing order of the CPU capacity they require, if placed on a leaf datacenter. Then, it places each service on a feasible datacenter that minimizes the contribution of this service to the overall cost in~\eqref{Eq:def_obj_func}. This benchmark is an adaptation to our problem of the CPVNF algorithm~\cite{CPVNF_proactive_place_in_CDN}, which was also used as a benchmark in~\cite{VNF_place_in_public_cloud}. 

{\em \ms~\cite{MultiScaler_Kassler}:}
It orders the critical services in non-decreasing order of the available CPU capacity in their current datacenter. It breaks ties by decreasing order of the CPU currently allocated for the service. Newly arriving requests are associated with the lowest priority, breaking ties by prioritizing services having fewer latency-feasible datacenters.
Each critical/new service is placed on the latency-feasible datacenter with the highest available capacity. 
This is an adaptation to our problem of the algorithm implied by (13)-(15) in~\cite{MultiScaler_Kassler}.

\textbf{Simulation methodology.}
We simulate users' mobility using SUMO~\cite{SUMO2018}. Benchmark algorithms utilize publicly available Python code~\cite{SFC_mig_Github}, while 
\alg\ is implemented using OMNeT++ network simulator~\cite{omnet}. 
Each new user, or an existing user which becomes critical, invokes a run of \alg, executed on the datacenter of the PoA serving the user. \alg's code is available in~\cite{dist_SFC_mig_Github}.

\subsection{Numerical results}\label{sec:sim:results}

{\bf Minimum required resources.}
We now study the computation resources each algorithm needs to allocate for the service requests arriving in each scenario, generated according to the user's mobility described by the considered real-world trace.  
A binary search is used to find the minimum amount of computation resources needed by each algorithm to successfully place all the critical services.
Clearly, a lower value of the required CPU implies better resource utilization.

Fig.~\ref{fig:cpu_vs_RT_prob} presents the results of this experiment, when varying the fraction of RT services.
The $y$ axis shows the overall CPU resources in the system. The required CPU resources consistently increase for higher fraction of RT services, as tighter latency constraints dictate allocating more resources closer to the edge. 
This phenomenon is especially noticeable in  Monaco, where edge resources are scarcer  (only 231 leaf datacenters in Monaco, compared to 1,524 in Luxembourg). 

For both scenarios, the amount of CPU BUPU requires is almost identical to LBound. Despite being fully distributed and asynchronous, the amount of CPU needed by \alg\ is only slightly higher than BUPU. 
A relatively low requirement of CPU capacity is also achieved by \ms, thanks to its clever prioritization of requests and datacenters, which effectively balances the processing load among the datacenters. 
Finally, \ffit\ and \cpvnf\ require processing capacities that are 50\% to 100\% higher than \lBound\ to provide a feasible solution.

\begin{figure*}
    \centering
    \subfloat[\label{fig:cpu_vs_RT_prob_Lux}Luxembourg] {
        \includegraphics[width=\subFigWidth]{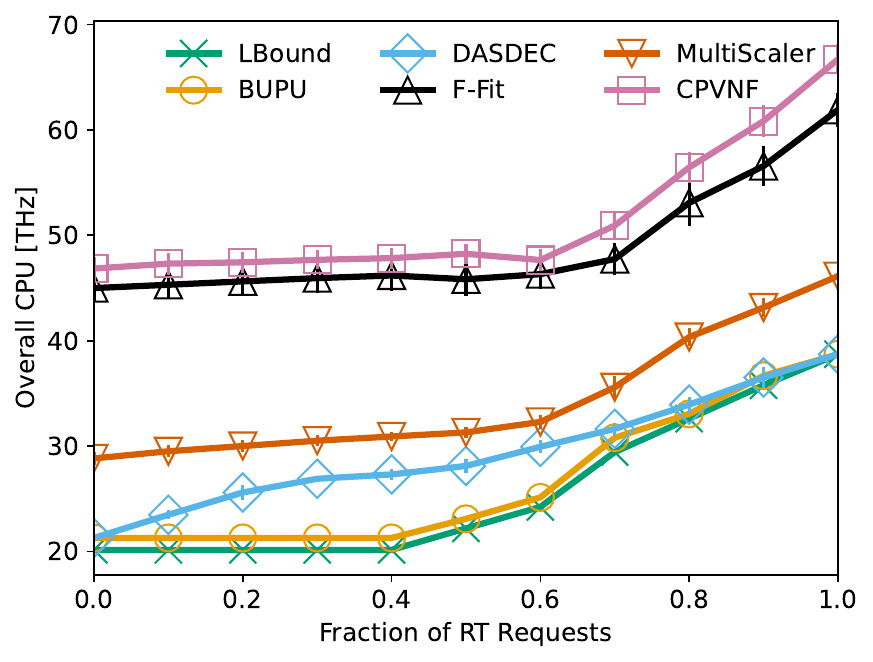}
    }
    \subfloat[\label{fig:cpu_vs_RT_prob_Monaco}Monaco] {
    \includegraphics[width=\subFigWidth]{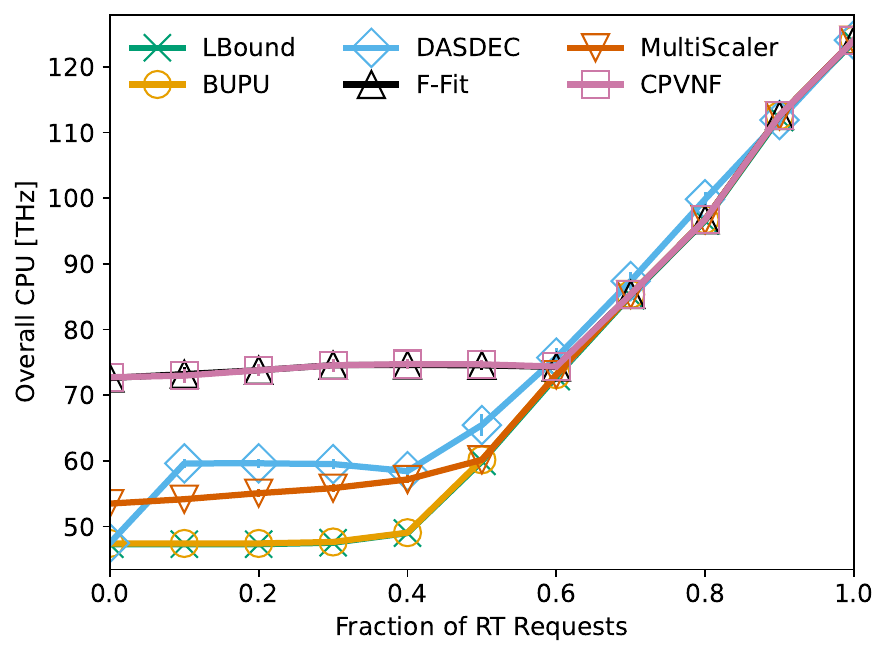}
    }
    \caption{Minimum required processing capacity when varying the ratio of RT service requests.}
    \label{fig:cpu_vs_RT_prob}
\end{figure*}

\begin{figure*}
														
    \centering
    \subfloat[\label{fig:ComOh_Lux}Luxembourg] {
        \includegraphics[width=\subFigWidth]{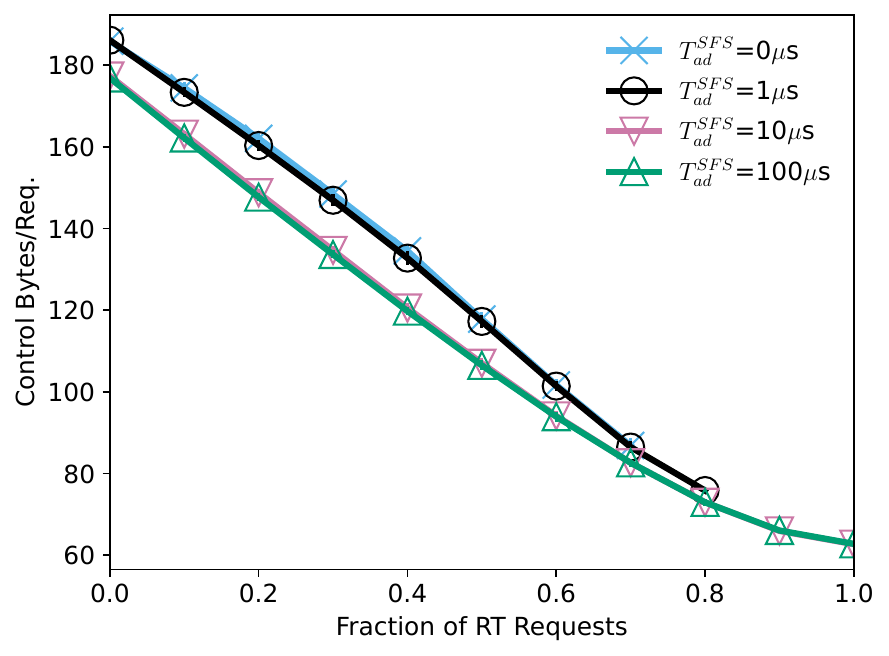}
    }
    \subfloat[\label{fig:ComOh_Monaco}Monaco] {
        \includegraphics[width=\subFigWidth]{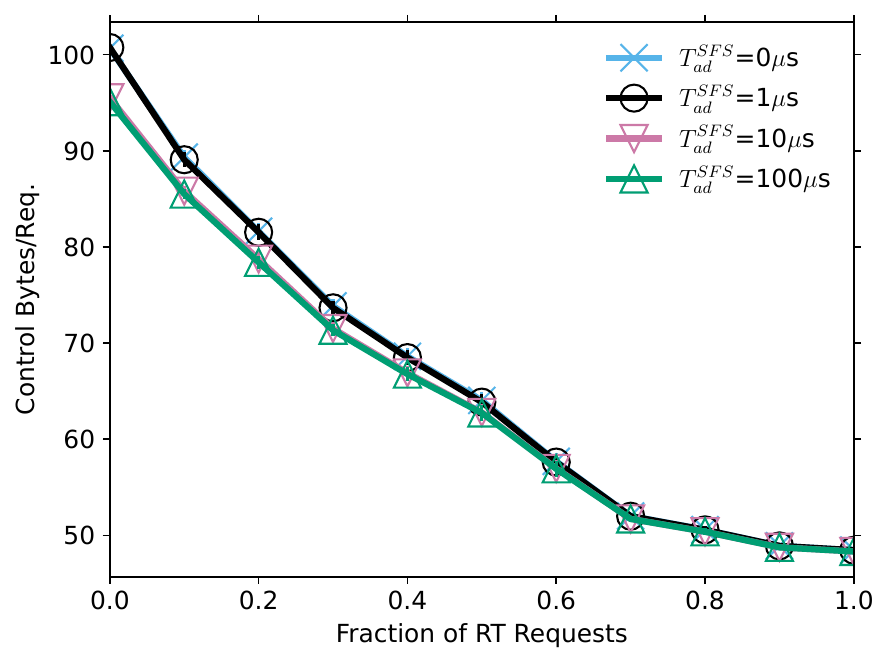}
    }
    \caption{Per-request signaling overhead due to \alg.}
    \label{fig:ComOh}
    \vspace{-0.1 cm}
\end{figure*}

{\bf Communication overhead.} 
The communication overhead is a key metric for the performance of any distributed algorithm. Hence, we now study the \alg's communication overhead for different scenarios. 

For each simulated scenario, we set the CPU resources to 10\% above the amount of resources required by LBound, to find a feasible solution when 100\% of the requests are RT. While maintaining this amount of CPU resources, we vary the ratio of RT requests and measure \alg's communication overhead. 
Our metric for the communication overhead is the overall byte count of all the messages exchanged by \alg\ (``control bytes'') during a mobility trace, over the overall number of critical/new requests occurring in that time period. 
Fig.~\ref{fig:ComOh} presents the per-request communication overhead for different values of the \bu\ accumulation delay parameter $\accDelayBu$.  

To set \pd's accumulation delay parameter $\accDelayPd$, we observe that a run of \pd() may also migrate non-critical services, thus incurring a higher overhead than \pu(). Hence, we set $\accDelayPd {=} 4 \cdot \accDelayBu$. 

The results show that increasing the fraction of RT requests decreases the signaling overhead. The reason is that RT requests can be placed only in the three lowest levels of the tree, thus avoiding any signaling message between higher-level datacenters. When the accumulation delays increase, the protocol aggregates more requests before sending a message, thus decreasing the signaling overhead. 
However, an accumulation delay of about $10~\mu\text{s}$ suffices. 
We stress that the accumulation delay only impacts the time until the protocol finds a new feasible placement, not the latency experienced by the user enjoying the service. Indeed, delaying the migration decision may deteriorate the user's QoE. However, this performance deterioration can be mitigated using an efficient prediction mechanism for the user's mobility. Furthermore, in practical scenarios, an accumulation delay of about $10~\mu\text{s}$ may be negligible compared to the latency incurred by the migration process.  
Finally, in all the considered settings, the signaling overhead associated with each request is around $100$~bytes, implying a negligible bandwidth overhead.

{\bf Cost comparison.} 
In the next experiment, we compare the cost of various solutions for the \migProb. We set the ratio of RT requests to 30\%, and vary the resource augmentation. 
Table~\ref{tab:cost_vs_rsrc} details the results of this comparison.  
 The CPU values $C_{cpu}$ are normalized with respect to the minimum amount of CPU required by \lBound\ to find a feasible solution $\hat{C}_{cpu}$ in the corresponding scenario. For a meaningful comparison, the cost at each row in the table is normalized w.r.t. the cost of \lBound's solution to the same settings. The table shows the normalized cost obtained by different algorithms when varying the normalized cost $C_{cpu}/\hat{C}_{cpu}$. For each scenario and algorithm, the table also details the minimal amount of CPU required in the scenario to find a feasible solution.
When no feasible solution is found, we define the cost as infinite.

\alg\ and BUPU achieve the lowest costs. BUPU, leveraging a centralized controller with a consistently fresh and accurate view of the whole system's state, successfully finds a feasible solution using resources that are only slightly higher than the lower bound. Even though \alg\ does not rely on a single centralized controller, it needs just slightly higher resources than BUPU to find a feasible solution. When resources are abundant (e.g., twice the minimal resources $\hat{C}_{cpu}$ required by LBound to find a feasible solution), most of the benchmark algorithms obtain costs that are almost identical to the lower bound. The highest costs are those obtained by \ms, which tends to favor load-balancing over cost reduction.

\begin{table}[tb!]
    \centering
{\scriptsize
    \caption{\label{tab:cost_vs_rsrc}{Normalized cost of service deployment and migration  vs.\  resource augmentation. The costs are normalized with respect to the lower bound. An infinite cost indicates that the algorithm cannot find a feasible solution}} 
        \begin{tabular}{|c|c|c|c|c|c|}
            \hline
            \multicolumn{6}{|c|}{Luxembourg}\\
            \hline
            \multirow{3}{*}{$\frac{C_{cpu}}{\hat{C}_{cpu}}$} &
            \multicolumn{5}{c|}{Normalized Cost}\\ 
            \cline{2-6} &
                    DAPP- & 
                    \multirow{2}{*}{BUPU} &
                    \multirow{2}{*}{\ms} &
                    \multirow{2}{*}{F-Fit} &
                    \multirow{2}{*}{CPVNF} \\
                    & ECC & & & &
                    \\ \hline
            1.06	& $\infty$  & 1.10	& $\infty$ & $\infty$ &	$\infty$ \\ \hline             
            1.33	& 1.08  	& 1.11  & $\infty$ & $\infty$ &	$\infty$ \\ \hline  
            1.54	& 1.05      & 1.06  & 1.27     & $\infty$ & $\infty$ \\ \hline 
            2.35	& 1.04      & 1.03	& 1.34     & 1.06     &	$\infty$ \\ \hline 
            2.40	& 1.04      & 1.03	& 1.34     & 1.06     &	1.06     \\ \hline 

            \multicolumn{6}{c}{}\\

            \hline
            \multicolumn{6}{|c|}{Monaco}\\
            \hline
            \multirow{3}{*}{$\frac{C_{cpu}}{\hat{C}_{cpu}}$} &
            \multicolumn{5}{c|}{Normalized Cost}\\ 
            \cline{2-6} &
                    DAPP- & 
                    \multirow{2}{*}{BUPU} &
                    \multirow{2}{*}{\ms} &
                    \multirow{2}{*}{F-Fit} &
                    \multirow{2}{*}{CPVNF} \\
                    & ECC & & & &
                    \\ \hline
            1.002	& $\infty$  & 1.07	& $\infty$ & $\infty$ &	$\infty$ \\ \hline             
            1.26	& 1.06  	& 1.06  & $\infty$ & $\infty$ &	$\infty$ \\ \hline  
            1.50	& 1.07      & 1.07	& 1.35     & $\infty$ & $\infty$ \\ \hline 
            1.58	& 1.08      & 1.09	& 1.36     & 1.10     &	1.10     \\ \hline 
            2.00	& 1.04      & 1.05	& 1.34     & 1.05	  & 1.05 	 \\ \hline 
        \end{tabular}           		
}
\end{table}

\section{Related Work}\label{sec:related_work}
State-of-the-art solutions to service deployment and migration~\cite{CPVNF_proactive_place_in_CDN, MoveWithMe, mig_correlated_VMs, dynamic_sched_and_reconf_t, SFC_mig, Dynamic_SFC_by_rtng_Dijkstra, Avatar, mig_or_reinstall, tong2016_justify_path_to_root, joint_placement_and_accs_netw}
assume a centralized orchestrator that possesses fresh, accurate information about all users' locations, trajectories, and computational demands, and the available resources at all the datacenters. 
Such an assumption may be impractical in a large system, possibly operated by several operators. This advocates the design of distributed approaches, as \alg.

Another approach is to let each service instance independently select an optimal destination based on factors such as topological distance, the availability of resources at the destination, and the data protection level at the destination~\cite{Companion_Fog, NFV_ego_learning, Dist_aleck_actually_ego_21}. This approach assumes that every service has a full fresh view of the available resources at each node. Such an assumption may be impractical in a large network, possibly operated by different providers. Furthermore, such a myopic strategy may fail to provide a feasible system-level solution when multiple RT users compete for the scarce resources on the edge.   

Several recent works~\cite{dist_SFC_aleck, orch_dist_openStack} refer to ``distributed placement'' but, in practice, focus on clustering, namely, improving the scalability by partitioning resources into small sets of datacenters on which each concrete service can be deployed.
The work~\cite{App_placement_in_fog_n_edge} uses dynamic clustering of datacenters to handle multiple simultaneous independent placement requests. However, the complex clustering mechanism may result in significant communication and computational overhead. Also,~\cite{App_placement_in_fog_n_edge} does not consider migrating non-critical requests to make room for a new one, as we do. 
    
Several works~\cite{Obl_vs_stateful_edge-cloud_Marco, Dist_edge-cloud_Marco, Edge_admission_cntrl_balk} consider the problem of directing a service either to an edge datacenter, or to the cloud. These papers use rigorous queuing-theory-based analysis to set the ratio of services directed to the cloud, where the objective is to minimize the service's completion time. The analysis highlights the potential benefits of a distributed approach, compared to the centralized approach, whose performance may significantly deteriorate in the face of wrong state estimation. However, these studies do not address user mobility and service migration, as considered in our work.

The \migProb\ combines several properties of the Multiple Knapsack problem~\cite{MultiKnap} with added restrictions typical for bin-packing problems (e.g., each item can be packed only on a subset of the knapsacks, and a feasible solution must pack all the items). In contrast to the usual settings of such problems, we aim at a distributed and asynchronous scheme that runs independently on multiple datacenters (``knapsacks''), using only little communication between them.  
APSR~\cite{APSR} considers a distributed randomized approach to place VM placement on a single cloud datacenter. 

A large body of work~\cite{Dance_elephants_VM_mig, Linux_Containers_mig_18, dynamic_sched_and_reconf_t, live_service_mig_layers, Kassler_SDN_mig18, Reconf_time, Mig_VNF_study} focuses on lower-level implementational details of the migration process, with the aim to decrease its overhead (service downtime, migration time, and quantity of resources consumed by migration). 
For instance,~\cite{Dance_elephants_VM_mig, Linux_Containers_mig_18, dynamic_sched_and_reconf_t} cleverly select which data to migrate and in what order and pace. 
The work~\cite{Multiple_mig_sched_Buyya} considers the bandwidth allocated to each migration process. 
Instead, implementation issues and performance overhead of virtual machine live migration are discussed in~\cite{Mig_VNF_study}. 
Some of these implementation-level details are orthogonal to the migration and placement protocol and, hence, could be incorporated into a practical implementation of \alg.

The prediction of future required migration typically relies on the users' trajectories and signal strengths~\cite{Dist_aleck_actually_ego_21}. 
The works~\cite{NFV_ego_learning, LSTM_predict_and_then_optimize, SARIMA} leverage learning algorithms to predict future service requests. \cite{LSTM_predict_and_then_optimize} further shows how distinct mobile operators can efficiently and securely cooperate to predict future requests. 
Many such prediction techniques can be incorporated into our solution, e.g., by giving the predicted future traffic and users' mobility as input to our placement protocol.

An early version of this work~\cite{Dist_SoftCom} sketched our algorithmic solution and presented a preliminary performance evaluation.
   
\section{Conclusions}\label{sec:conclusion}

We proposed a distributed asynchronous solution framework for service provisioning in edge-cloud multi-tier systems. 
As a distributed and asynchronous approach, our solution is well-suited for computing infrastructures managed by multiple operators and scales effectively with the number of datacenters.
Furthermore, our algorithm converges after exchanging a bounded number of control messages.
Numerical results from realistic settings demonstrate that our approach achieves placement with minimal additional computing resources compared to centralized solutions, which may not scale well in large systems.
Importantly, we showed that the required communication overhead is negligible.

\bibliographystyle{IEEEtran}
\bibliography{Refs}

\begin{thebibliography}{10}
\providecommand{\url}[1]{#1}
\csname url@samestyle\endcsname
\providecommand{\newblock}{\relax}
\providecommand{\bibinfo}[2]{#2}
\providecommand{\BIBentrySTDinterwordspacing}{\spaceskip=0pt\relax}
\providecommand{\BIBentryALTinterwordstretchfactor}{4}
\providecommand{\BIBentryALTinterwordspacing}{\spaceskip=\fontdimen2\font plus
\BIBentryALTinterwordstretchfactor\fontdimen3\font minus
  \fontdimen4\font\relax}
\providecommand{\BIBforeignlanguage}[2]{{%
\expandafter\ifx\csname l@#1\endcsname\relax
\typeout{** WARNING: IEEEtran.bst: No hyphenation pattern has been}%
\typeout{** loaded for the language `#1'. Using the pattern for}%
\typeout{** the default language instead.}%
\else
\language=\csname l@#1\endcsname
\fi
#2}}
\providecommand{\BIBdecl}{\relax}
\BIBdecl

\bibitem{FollowMe_J}
T.~Taleb, A.~Ksentini, and P.~A. Frangoudis, ``Follow-me cloud: When cloud
  services follow mobile users,'' \emph{IEEE Transactions on Cloud Computing},
  vol.~7, no.~2, pp. 369--382, 2016.

\bibitem{micado_orchestrator}
A.~Ullah, H.~Dagdeviren, R.~C. Ariyattu, J.~DesLauriers, T.~Kiss, and
  J.~Bowden, ``Micado-edge: Towards an application-level orchestrator for the
  cloud-to-edge computing continuum,'' \emph{Journal of Grid Computing},
  vol.~19, no.~4, pp. 1--28, 2021.

\bibitem{SFC_mig}
D.~Zhao, G.~Sun, D.~Liao, S.~Xu, and V.~Chang, ``Mobile-aware service function
  chain migration in cloud--fog computing,'' \emph{Future Generation Computer
  Systems}, vol.~96, pp. 591--604, 2019.

\bibitem{orch_cloud2edge_survey}
S.~Svorobej, M.~Bendechache, F.~Griesinger, and J.~Domaschka, ``Orchestration
  from the cloud to the edge,'' \emph{The Cloud-to-Thing Continuum}, pp.
  61--77, 2020.

\bibitem{MoveWithMe}
R.~Bruschi, F.~Davoli, P.~Lago, and J.~F. Pajo, ``Move with me: Scalably
  keeping virtual objects close to users on the move,'' in \emph{IEEE ICC},
  2018, pp. 1--6.

\bibitem{Justify_CLP_tree_ISPs}
Y.-D. Lin, C.-C. Wang, C.-Y. Huang, and Y.-C. Lai, ``Hierarchical cord for
  {NFV} datacenters: resource allocation with cost-latency tradeoff,''
  \emph{IEEE Network}, vol.~32, no.~5, pp. 124--130, 2018.

\bibitem{tong2016_justify_path_to_root}
L.~Tong, Y.~Li, and W.~Gao, ``A hierarchical edge cloud architecture for mobile
  computing,'' in \emph{IEEE INFOCOM}, 2016, pp. 1--9.

\bibitem{mig_correlated_VMs}
G.~Sun, D.~Liao, D.~Zhao, Z.~Xu, and H.~Yu, ``Live migration for multiple
  correlated virtual machines in cloud-based data centers,'' \emph{IEEE
  Transactions on Services Computing}, pp. 279--291, 2015.

\bibitem{MoveWithMe_short}
R.~Bruschi \emph{et~al.}, ``Move with me: Scalably keeping virtual objects
  close to users on the move,'' in \emph{IEEE ICC}, 2018, pp. 1--6.

\bibitem{Justifies_path_to_root_n_CLP_vehs}
B.~Kar, K.-M. Shieh, Y.-C. Lai, Y.-D. Lin, and H.-W. Ferng, ``{QoS} violation
  probability minimization in federating vehicular-fogs with cloud and edge
  systems,'' \emph{IEEE Transactions on Vehicular Technology}, vol.~70, no.~12,
  pp. 13\,270--13\,280, 2021.

\bibitem{Dynamic_Service_Provisioning_ToN}
I.~Cohen, C.~F. Chiasserini, P.~Giaccone, and G.~Scalosub, ``Dynamic service
  provisioning in the edge-cloud continuum with bounded resources,'' \emph{IEEE
  Transactions on Networking}, 2023.

\bibitem{CPVNF_proactive_place_in_CDN}
M.~Dieye, S.~Ahvar, J.~Sahoo, E.~Ahvar, R.~Glitho, H.~Elbiaze, and N.~Crespi,
  ``{CPVNF}: Cost-efficient proactive {VNF} placement and chaining for
  value-added services in content delivery networks,'' \emph{IEEE Transactions
  on Network and Service Management}, vol.~15, no.~2, pp. 774--786, 2018.

\bibitem{SFC_mig_mechanism}
H.~Yu, J.~Yang, and C.~Fung, ``Elastic network service chain with fine-grained
  vertical scaling,'' in \emph{IEEE GLOBECOM}, 2018, pp. 1--7.

\bibitem{MultiScaler_Kassler}
A.~Al-Dulaimy \emph{et~al.}, ``{MULTISCALER}: A multi-loop auto-scaling
  approach for cloud-based applications,'' \emph{IEEE TCC}, 2020.

\bibitem{dynamic_sched_and_reconf_t}
I.~Leyva-Pupo, C.~Cervell{\'o}-Pastor, C.~Anagnostopoulos, and D.~P. Pezaros,
  ``Dynamic scheduling and optimal reconfiguration of {UPF} placement in {5G}
  networks,'' in \emph{ACM MSWiM}, 2020, pp. 103--111.

\bibitem{agarwal-19}
S.~Agarwal, F.~Malandrino, C.~F. Chiasserini, and S.~De, ``Vnf placement and
  resource allocation for the support of vertical services in 5g networks,''
  \emph{IEEE Transactions on Networking}, vol.~27, no.~1, pp. 433--446, 2019.

\bibitem{Avatar}
X.~Sun and N.~Ansari, ``{PRIMAL}: Profit maximization avatar placement for
  mobile edge computing,'' in \emph{IEEE ICC}, 2016, pp. 1--6.

\bibitem{Dynamic_SFC_by_rtng_Dijkstra}
T.~Mahboob, Y.~R. Jung, and M.~Y. Chung, ``Dynamic {VNF} placement to manage
  user traffic flow in software-defined wireless networks,'' \emph{Journal of
  Network and Systems Management, Springer}, pp. 1--21, 2020.

\bibitem{APSR}
I.~Cohen, G.~Einziger, M.~Goldstein, Y.~Sa’ar, G.~Scalosub, and E.~Waisbard,
  ``High throughput vms placement with constrained communication overhead and
  provable guarantees,'' \emph{IEEE Transactions on Network and Service
  Management}, 2023.

\bibitem{Obl_vs_stateful_edge-cloud_Marco}
V.~Mancuso, P.~Castagno, M.~Sereno, and M.~A. Marsan, ``Stateful versus
  stateless selection of edge or cloud servers under latency constraints,'' in
  \emph{IEEE WoWMoM}, 2022, pp. 110--119.

\bibitem{Crosshaul}
A.~De~La~Oliva \emph{et~al.}, ``Final 5g-crosshaul system design and economic
  analysis,'' \emph{5G-Crosshaul public deliverable}, 2017.

\bibitem{NFV_ego_learning}
T.~Ouyang \emph{et~al.}, ``Adaptive user-managed service placement for mobile
  edge computing: An online learning approach,'' in \emph{IEEE INFOCOM}, 2019,
  pp. 1468--1476.

\bibitem{Mig_in_Mobile_Edge_Clouds}
S.~Wang, R.~Urgaonkar, M.~Zafer, T.~He, K.~Chan, and K.~K. Leung, ``Dynamic
  service migration in mobile edge-clouds,'' in \emph{IEEE IFIP Networking},
  2015, pp. 1--9.

\bibitem{Okpi}
J.~Martín-Pérez, F.~Malandrino, C.~F. Chiasserini, M.~Groshev, and C.~J.
  Bernardos, ``Multiagent graph coloring: Pareto efficiency, fairness and
  individual rationality,'' in \emph{KPI Guarantees in Network Slicing},
  vol.~30, no.~2, 2022, pp. 655--668.

\bibitem{Dynamic_user_demands}
M.~Nguyen, M.~Dolati, and M.~Ghaderi, ``Deadline-aware {SFC} orchestration
  under demand uncertainty,'' \emph{IEEE TNSM}, pp. 2275--2290, 2020.

\bibitem{Book_GAP_and_even_knowing_if_exists_feas_is_NPH}
S.~Martello and P.~Toth, \emph{Knapsack problems: algorithms and computer
  implementations}.\hskip 1em plus 0.5em minus 0.4em\relax John Wiley \& Sons,
  Inc., 1990.

\bibitem{Luxembourg}
L.~Codec{\'a}, R.~Frank, S.~Faye, and T.~Engel, ``Luxembourg {SUMO} traffic
  ({LuST}) scenario: Traffic demand evaluation,'' \emph{IEEE Intelligent
  Transportation Systems Magazine}, vol.~9, no.~2, pp. 52--63, 2017.

\bibitem{Monaco}
L.~Codeca and J.~H{\"a}rri, ``Monaco {SUMO} traffic ({MoST}) scenario: A {3D}
  mobility scenario for cooperative {ITS},'' \emph{EPiC Series in Engineering},
  vol.~2, pp. 43--55, 2018.

\bibitem{OpenCellid}
``Opencellid,'' https://opencellid.org/, accessed on 9.9.2024.

\bibitem{ca}
G.~Avino \emph{et~al.}, ``A {MEC}-based extended virtual sensing for automotive
  services,'' \emph{IEEE Transactions on Network and Service Management},
  vol.~16, no.~4, pp. 1450--1463, 2019.

\bibitem{see-through}
F.~Rameau, H.~Ha, K.~Joo, J.~Choi, K.~Park, and I.~S. Kweon, ``A real-time
  augmented reality system to see-through cars,'' \emph{IEEE Transactions on
  Visualization and Computer Graphics}, pp. 2395--2404, 2016.

\bibitem{SFC_mig_short}
D.~Zhao \emph{et~al.}, ``Mobile-aware service function chain migration in
  cloud--fog computing,'' \emph{Future Generation Computer Systems}, vol.~96,
  pp. 591--604, 2019.

\bibitem{Gurobi}
\BIBentryALTinterwordspacing
``Gurobi optimizer,'' accessed on 9.9.2024. [Online]. Available:
  \url{https://www.gurobi.com}
\BIBentrySTDinterwordspacing

\bibitem{App_placement_in_fog_n_edge}
M.~Goudarzi, M.~Palaniswami, and R.~Buyya, ``A distributed application
  placement and migration management techniques for edge and fog computing
  environments,'' in \emph{IEEE FedCSIS}, 2021, pp. 37--56.

\bibitem{VNF_place_in_public_cloud}
T.~Gao \emph{et~al.}, ``Cost-efficient {VNF} placement and scheduling in public
  cloud networks,'' \emph{IEEE Transactions on Communications}, pp. 4946--4959,
  2020.

\bibitem{SUMO2018}
P.~Alvarez \emph{et~al.}, ``Microscopic traffic simulation using sumo,'' in
  \emph{IEEE International Conference on Intelligent Transportation Systems},
  2018.

\bibitem{SFC_mig_Github}
\BIBentryALTinterwordspacing
``Service function chains migration.'' [Online]. Available:
  \url{https://github.com/ofanan/SFC_migration}
\BIBentrySTDinterwordspacing

\bibitem{omnet}
\BIBentryALTinterwordspacing
``{OMNeT++} discrete event simulator.'' [Online]. Available:
  \url{https://omnetpp.org}
\BIBentrySTDinterwordspacing

\bibitem{dist_SFC_mig_Github}
\BIBentryALTinterwordspacing
``Distributed {SFC} migration.'' [Online]. Available:
  \url{https://github.com/ofanan/Distributed_SFC_migration}
\BIBentrySTDinterwordspacing

\bibitem{mig_or_reinstall}
H.~Hawilo, M.~Jammal, and A.~Shami, ``Orchestrating network function
  virtualization platform: Migration or re-instantiation?'' in \emph{IEEE
  CloudNet}, 2017, pp. 1--6.

\bibitem{joint_placement_and_accs_netw}
B.~Gao, Z.~Zhou, F.~Liu, F.~Xu, and B.~Li, ``An online framework for joint
  network selection and service placement in mobile edge computing,''
  \emph{IEEE Transactions on Mobile Computing}, vol.~21, no.~11, pp.
  3836--3851, 2021.

\bibitem{Companion_Fog}
C.~Puliafito, E.~Mingozzi, C.~Vallati, F.~Longo, and G.~Merlino, ``Companion
  fog computing: Supporting things mobility through container migration at the
  edge,'' in \emph{IEEE SMARTCOMP}, 2018, pp. 97--105.

\bibitem{Dist_aleck_actually_ego_21}
L.~Pacheco, D.~Ros{\'a}rio, E.~Cerqueira, L.~Villas, T.~Braun, and A.~A.
  Loureiro, ``Distributed user-centric service migration for edge-enabled
  networks,'' in \emph{2021 IFIP/IEEE International Symposium on Integrated
  Network Management (IM)}.\hskip 1em plus 0.5em minus 0.4em\relax IEEE, 2021,
  pp. 618--622.

\bibitem{dist_SFC_aleck}
M.~Ghaznavi, N.~Shahriar, S.~Kamali, R.~Ahmed, and R.~Boutaba, ``Distributed
  service function chaining,'' \emph{IEEE Journal on Selected Areas in
  Communications}, vol.~35, no.~11, pp. 2479--2489, 2017.

\bibitem{orch_dist_openStack}
D.~Haja \emph{et~al.}, ``How to orchestrate a distributed openstack,'' in
  \emph{IEEE INFOCOM WKSHPS}, 2018, pp. 293--298.

\bibitem{Dist_edge-cloud_Marco}
V.~Mancuso, L.~Badia, P.~Castagno, M.~Sereno, and M.~A. Marsan, ``Efficiency of
  distributed selection of edge or cloud servers under latency constraints,''
  in \emph{MedComNet}.\hskip 1em plus 0.5em minus 0.4em\relax IEEE, 2023, pp.
  158--166.

\bibitem{Edge_admission_cntrl_balk}
S.~Chen, L.~Wang, and F.~Liu, ``Optimal admission control mechanism design for
  time-sensitive services in edge computing,'' in \emph{IEEE INFOCOM}, 2022,
  pp. 1169--1178.

\bibitem{MultiKnap}
M.~S. Hung and J.~C. Fisk, ``An algorithm for 0-1 multiple-knapsack problems,''
  \emph{Naval Research Logistics Quarterly}, vol.~25, no.~3, pp. 571--579,
  1978.

\bibitem{Dance_elephants_VM_mig}
K.~Ha \emph{et~al.}, ``You can teach elephants to dance: Agile {VM} handoff for
  edge computing,'' in \emph{ACM/IEEE SEC}, 2017, pp. 1--14.

\bibitem{Linux_Containers_mig_18}
R.~Stoyanov and M.~J. Kollingbaum, ``Efficient live migration of linux
  containers,'' in \emph{ISC High Performance}.\hskip 1em plus 0.5em minus
  0.4em\relax Springer, 2018, pp. 184--193.

\bibitem{live_service_mig_layers}
A.~Machen, S.~Wang, K.~K. Leung, B.~J. Ko, and T.~Salonidis, ``Live service
  migration in mobile edge clouds,'' \emph{IEEE Wireless Communications}, pp.
  140--147, 2017.

\bibitem{Kassler_SDN_mig18}
K.~A. Noghani, A.~Kassler, and P.~S. Gopannan, ``{EVPN/SDN} assisted live {VM}
  migration between geo-distributed data centers,'' in \emph{IEEE NetSoft},
  2018, pp. 105--113.

\bibitem{Reconf_time}
R.~Cziva, C.~Anagnostopoulos, and D.~P. Pezaros, ``Dynamic, latency-optimal
  {VNF} placement at the network edge,'' in \emph{IEEE INFOCOM}, 2018, pp.
  693--701.

\bibitem{Mig_VNF_study}
S.~Ramanathan, K.~Kondepu, M.~Razo, M.~Tacca, L.~Valcarenghi, and A.~Fumagalli,
  ``Live migration of virtual machine and container based mobile core network
  components: A comprehensive study,'' \emph{IEEE Access}, vol.~9, pp.
  105\,082--105\,100, 2021.

\bibitem{Multiple_mig_sched_Buyya}
T.~He, A.~N. Toosi, and R.~Buyya, ``{SLA}-aware multiple migration planning and
  scheduling in {SDN-NFV}-enabled clouds,'' \emph{Journal of Systems and
  Software}, vol. 176, p. 110943, 2021.

\bibitem{LSTM_predict_and_then_optimize}
T.~Subramanya and R.~Riggio, ``Centralized and federated learning for
  predictive {VNF} autoscaling in multi-domain {5G} networks and beyond,''
  \emph{IEEE TNSM}, vol.~18, no.~1, pp. 63--78, 2021.

\bibitem{SARIMA}
V.~Eramo \emph{et~al.}, ``Reconfiguration of optical-nfv network architectures
  based on cloud resource allocation and qos degradation cost-aware prediction
  techniques,'' \emph{IEEE Access}, vol.~8, pp. 200\,834--200\,850, 2020.

\bibitem{Dist_SoftCom}
I.~Cohen, P.~Giaccone, and C.~F. Chiasserini, ``Distributed asynchronous
  protocol for service provisioning in the edge-cloud continuum,'' in
  \emph{IEEE SoftCom}, 2023, pp. 1--6.

\end{thebibliography}

\begin{IEEEbiographynophoto}{Itamar Cohen}
(Member, IEEE) is currently a Faculty Member at the School of Computer Science, Ariel University at Ariel, 40700, Israel. His research interests include scheduling, cooperative caching, and decision-making under uncertainty. 
\end{IEEEbiographynophoto}
\begin{IEEEbiographynophoto}{Paolo Giaccone}
(Senior Member, IEEE) is currently a Professor with the
Department of Electronics and Telecommunications, Politecnico di Torino,
Italy. His main research interests include the design of network control and
optimization algorithms.
\end{IEEEbiographynophoto}
\begin{IEEEbiographynophoto}{ Carla Fabiana Chiasserini}
 (Fellow, IEEE) was a Visiting Researcher with the University of California, San Diego (UCSD). She was a Visiting Professor with Monash University from 2012 to 2016 and Technische Universität Berlin (TUB) from 2021 to 2022. She is currently a Professor with Politecnico di Torino.
\end{IEEEbiographynophoto}

\end{document}